\documentclass[a4paper,11pt,reqno]{amsart}
\usepackage{amsmath,graphicx,graphics,amssymb,psfrag,mathdots,mathtools,ulem,xcolor}
\usepackage[colorlinks=true]{hyperref}
\newtheorem{theo}{Theorem}[section]
\newtheorem{lem}[theo]{Lemma}
\newtheorem{con}[theo]{Conjecture}

\newtheorem{prop}[theo]{Proposition}

\theoremstyle{remark}

\numberwithin{equation}{section}

\newcommand{\M}{\operatorname{M}}

\newcommand{\al}{\alpha}
\newcommand{\be}{\beta}

\newcommand{\z}{\zeta}
\newcommand{\ch}{\operatorname{q}}

\newcommand{\Z}{\mathbb Z}
\newcommand{\Q}{\mathbb Q}
\newcommand{\C}{\mathbb C}

\newcommand{\q}{\operatorname{q}}
\newcommand{\proj}{\operatorname{proj}}
\newcommand{\dee}{\operatorname{d}}

\renewcommand{\mod}{\operatorname{mod}}

\newmuskip\pFqskip
\mathchardef\pFcomma=\mathcode`, 

\begin{document}

\title[The effect of microscopic gap displacement on the correlation of gaps]{The effect of microscopic gap displacement \\on the correlation of gaps in dimer systems}

\author[Mihai Ciucu]{\box\Adr}

\newbox\Adr
\setbox\Adr\vbox{ 
\centerline{ \large Mihai Ciucu} \vspace{0.3cm}
\centerline{Indiana University, Department of Mathematics}
\centerline{Bloomington, IN 47401, USA}
\centerline{{\tt mciucu@indiana.edu}} 
\vspace{0.5cm}
}

\thanks{Supported in part by National Science Foundation, DMS grant 1501052.}

\begin{abstract} In earlier work we showed that in the bulk, the correlation of gaps in dimer systems on the hexagonal lattice is governed, in the fine mesh limit, by Coulomb's law for 2D electrostatics. We also proved that the scaling limit of the discrete field ${\bold F}$ of average tile orientations is, up to a multiplicative constant, the electric field produced by a 2D system of charges corresponding to the gaps. In this paper we show that in the bulk, the relative change $T_{\al,\be}$ in correlation caused by displacing a hole by a fixed vector $(\al,\be)$ is, in the fine mesh limit, the projection on $(\al,\be)$ of a new field ${\bold T}$, which is also equal up to a multiplicative constant to the electric field of the corresponding system of charges.
We also discuss the differences between the fields ${\bold T}$ and ${\bold F}$ and present conjectures for their fine mesh limits in the more general case of a dimer system with boundary. The new field ${\bold T}$ can be viewed as capturing the instantaneous pull on each gap in the surrounding fluctuating sea of dimers. From the point of view of the parallel to physics, the electrostatic force emerges then as an entropic force.

\end{abstract}


\maketitle

\section{Introduction}

Consider a large region on the triangular lattice, say the region illustrated in Figure \ref{fjb}, with some fixed collection of holes. We want to consider lozenge tilings of such regions with holes, with two kinds of boundary conditions: Along the zig-zag portions of the boundary (shown in solid lines in Figure~\ref{fjb}) the lozenges are constrained to be inside the region, while along the straight lattice line portions (shown in dashed lines) lozenges are allowed to protrude out halfway.

Two natural questions that arise are:

\medskip
(1) How does the number of lozenge tilings\footnote{ A lozenge is the union of two unit triangles that share an edge. A lozenge tiling of a lattice region on the triangular lattice is a covering of the region by lozenges with no gaps or overlaps.} change for a fixed region if the holes are moved around?

\smallskip
(2) What is the average orientation\footnote{ When averaged over all lozenge tilings.} of the lozenge covering any given left-pointing unit triangle in the region?

\medskip
%
%
For the case when the region has no boundary (i.e. for lozenge tilings of the plane with a finite collection of holes), question (1) --- with ``moved around'' interpreted as displacing the holes {\it macroscopically} --- was answered in \cite{sc} and \cite{ec}, where we showed that for quite general collections of holes their correlation (a quantity defined as a limit of lozenge tilings; see Section 2) is governed in the scaling limit by the Coulomb energy of a 2D system of electrical charges naturally corresponding to the holes\footnote{ This system is obtained by replacing each hole by an electrical point charge of magnitude equal to the number of right-pointing unit triangles in the hole minus the number of left-pointing unit triangles in the hole.}. We conjectured in \cite{ov} that this holds for arbitrary collections of holes. Dub\'edat proved the square lattice counterpart of this conjecture in \cite{Dub} for the special case of unit holes. We proved the case of holes of arbitrary magnitude lined up along a diagonal of the square lattice in \cite{gd}. In \cite{spg} we extended these results to arbitrary planar, bipartite, weighted 2-periodic graphs in the liquid phase of the Kenyon-Okounkov-Sheffield classification of dimer models \cite{KOS}. 

For a half-plane with zig-zag boundary and a quite general collection of triangular holes of side-length two, question (1) was answered in \cite{sc}. The case of a half-plane with free lattice line boundary and a single triangular hole of side two was worked out in \cite{free}. We answered the square lattice analog of question (1) for the case of Aztec rectangles with holes along a symmetry axis in \cite{fin}. Further examples were worked out in \cite{angle} and \cite{rangle}. In the latter we also gave a conjectural answer for the general case of question (1) in the scaling limit.

For the case of no boundary, question (2) was answered in \cite{ef}, where we showed (for quite general collections of holes) that in the scaling limit the lozenges line up along the lines of the 2D electric field of the naturally corresponding system of charges mentioned above.


In this paper we consider the following finer version of question (1):

\medskip
(1'): What is the relative change in the number of tilings when a hole is displaced {\it microscopically}?

\medskip
We give the answer to question (1') in the case of no boundary, for a fairly general collection of holes (see Figure \ref{fbb} for an illustrative example), in Theorem \ref{tba}, which states that in the scaling limit the relative change in the number of tilings caused by microscopic displacements of a hole is governed by the logarithmic gradient of the Coulomb energy. We also give a conjectural answer for the general case with boundary in Conjecture \ref{tjb}. It turns out that this answer depends on the answer to question~(2), which we phrase in Conjecture \ref{tja}.




We can interpret the statement of Theorem \ref{tba} as follows. Suppose a hole moves to one of the six nearest possible positions (while the other holes are kept fixed), with probabilities proportional to the number of tilings compatible with the new position. Then the average displacement vector is lined up, in the scaling limit, along the electric field produced by the charges corresponding to the other holes, measured at the location of the moving hole. From this point of view, the electrostatic force emerges as an entropic force.
%
%
We discuss this in detail in Section 8.

Counting tilings of a lattice region in the plane --- equivalently, counting perfect matchings (also called dimer coverings) of its planar dual --- is a classical problem in combinatorics. The number of lozenge tilings of a hexagon (in the equivalent language of plane partitions; see \cite{DT}) was determined by MacMahon \cite{MacM}. A product formula for the number of domino tilings of a rectangular region on the square lattice was found independently by Kasteleyn \cite{Kast} and by Temperley and Fisher \cite{TF,FisherDimer}. The concept of correlation of gaps in a sea of dimers was introduced by Fisher and Stephenson in \cite{FS}. In related work, Zuber and Itzykson \cite{Zuber} studied the correlation of spins in the Ising model on the square lattice. For an overview of connections with the Coulomb gas model, see the survey \cite{Nienhuis} by Nienhuis. In \cite{KeLocal} Kenyon gave a practical way to compute the correlation of certain gaps in dimer systems as a determinant. In \cite{KOS} Kenyon, Okounkov and Sheffield gave a general classification of dimer models on planar, bipartite, weighted 2-periodic lattices. Another point of contact with the previous literature is the work by Baik, Kriecherbauer, McLaughlin and Miller \cite{BaikOne,BaikTwo} on the correlation of collinear edges in large hexagonal regions on the hexagonal lattice.

\section{Definitions and statement of results}
















The unit triangles of the triangular lattice can be conveniently coordinatized as follows. Each unit triangle has one vertical side. Each vertical unit segment in the lattice is contained in a unique right-pointing unit triangle, and in a unique left-pointing unit triangle. Mark the midpoints of the vertical unit segments, and coordinatize the marked points by a 60 degree oblique system of coordinates as indicated in Figure \ref{fba}. Then each right-pointing unit triangle gets the coordinates of the midpoint of its vertical side; left-pointing unit triangles are coordinatized the same way.

We denote by 
${\begin{matrix} \triangleright &  \\[-10.45pt] \triangleleft & \hskip-0.233in\triangleright \\[-10.45pt] \triangleright & \end{matrix}}\!\!\!_{a,b}$
the right-pointing triangular hole of side length two, placed in the plane so that its central unit triangle has coordinates $(a,b)$; 
${\begin{matrix} & \hskip-0.153in \triangleleft \\[-10.45pt]
\triangleleft & \hskip-0.153in\triangleright \\[-10.45pt]
& \hskip-0.153in\triangleleft \end{matrix}}_{a,b}$ has the obvious analogous meaning.

For $q\in\Q$ and a strictly increasing list of integers ${\bold a}=(a_1,\dotsc,a_s)$ for which $qa_i\in\Z$ and the ${\begin{matrix} \triangleright &  \\[-10.45pt] \triangleleft & \hskip-0.233in\triangleright \\[-10.45pt] \triangleright & \end{matrix}}\!\!\!_{a_i,qa_i}$'s (equivalently, the
${\begin{matrix} & \hskip-0.153in \triangleleft \\[-10.45pt]
\triangleleft & \hskip-0.153in\triangleright \\[-10.45pt]
& \hskip-0.153in\triangleleft \end{matrix}}_{a_i,qa_i}$'s)
are mutually disjoint, $i=1,\dotsc,s$, define the {\it linear multiholes} 
${\begin{matrix} \triangleright &  \\[-10.45pt] \triangleleft & \hskip-0.233in\triangleright \\[-10.45pt] \triangleright & \end{matrix}}_{\!\!\!\!\!\bold a}^{\!\!\!\!\!q}$
and
${\begin{matrix} & \hskip-0.153in \triangleleft \\[-10.45pt]
\triangleleft & \hskip-0.153in\triangleright \\[-10.45pt]
& \hskip-0.153in\triangleleft \end{matrix}}_{\,\bold a}^{\,q}$
(called linear multiholes of {\it slope} $q$) by
\begin{align*}
{\begin{matrix} \triangleright &  \\[-10.45pt] \triangleleft & \hskip-0.233in\triangleright \\[-10.45pt] \triangleright & \end{matrix}}_{\!\!\!\!\!\bold a}^{\!\!\!\!\!q}
={\begin{matrix} \triangleright &  \\[-10.45pt] \triangleleft & \hskip-0.233in\triangleright \\[-10.45pt] \triangleright & \end{matrix}}\!\!\!_{a_1,qa_1}
\cup\cdots\cup
{\begin{matrix} \triangleright &  \\[-10.45pt] \triangleleft & \hskip-0.233in\triangleright \\[-10.45pt] \triangleright & \end{matrix}}\!\!\!_{a_s,qa_s}
\\[10pt]
{\begin{matrix} & \hskip-0.153in \triangleleft \\[-10.45pt]
\triangleleft & \hskip-0.153in\triangleright \\[-10.45pt]
& \hskip-0.153in\triangleleft \end{matrix}}_{\,\bold a}^{\,q}
=
{\begin{matrix} & \hskip-0.153in \triangleleft \\[-10.45pt]
\triangleleft & \hskip-0.153in\triangleright \\[-10.45pt]
& \hskip-0.153in\triangleleft \end{matrix}}_{a_1,qa_1}
\cup\cdots\cup
{\begin{matrix} & \hskip-0.153in \triangleleft \\[-10.45pt]
\triangleleft & \hskip-0.153in\triangleright \\[-10.45pt]
& \hskip-0.153in\triangleleft \end{matrix}}_{a_s,qa_s}.
\end{align*}
Figure \ref{fbb} shows several examples of linear multiholes.

\begin{figure}[t]
\centerline{
\hfill
{\includegraphics[width=0.30\textwidth]{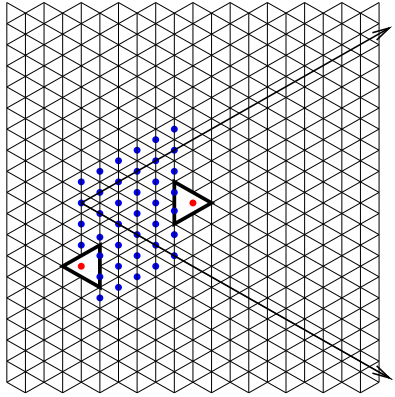}}
\hfill
}
\caption{\label{fba} The $60^\circ$ coordinate system and two triangular holes of side 2.}
\end{figure}

\begin{figure}[h]
\centerline{
\hfill
{\includegraphics[width=0.70\textwidth]{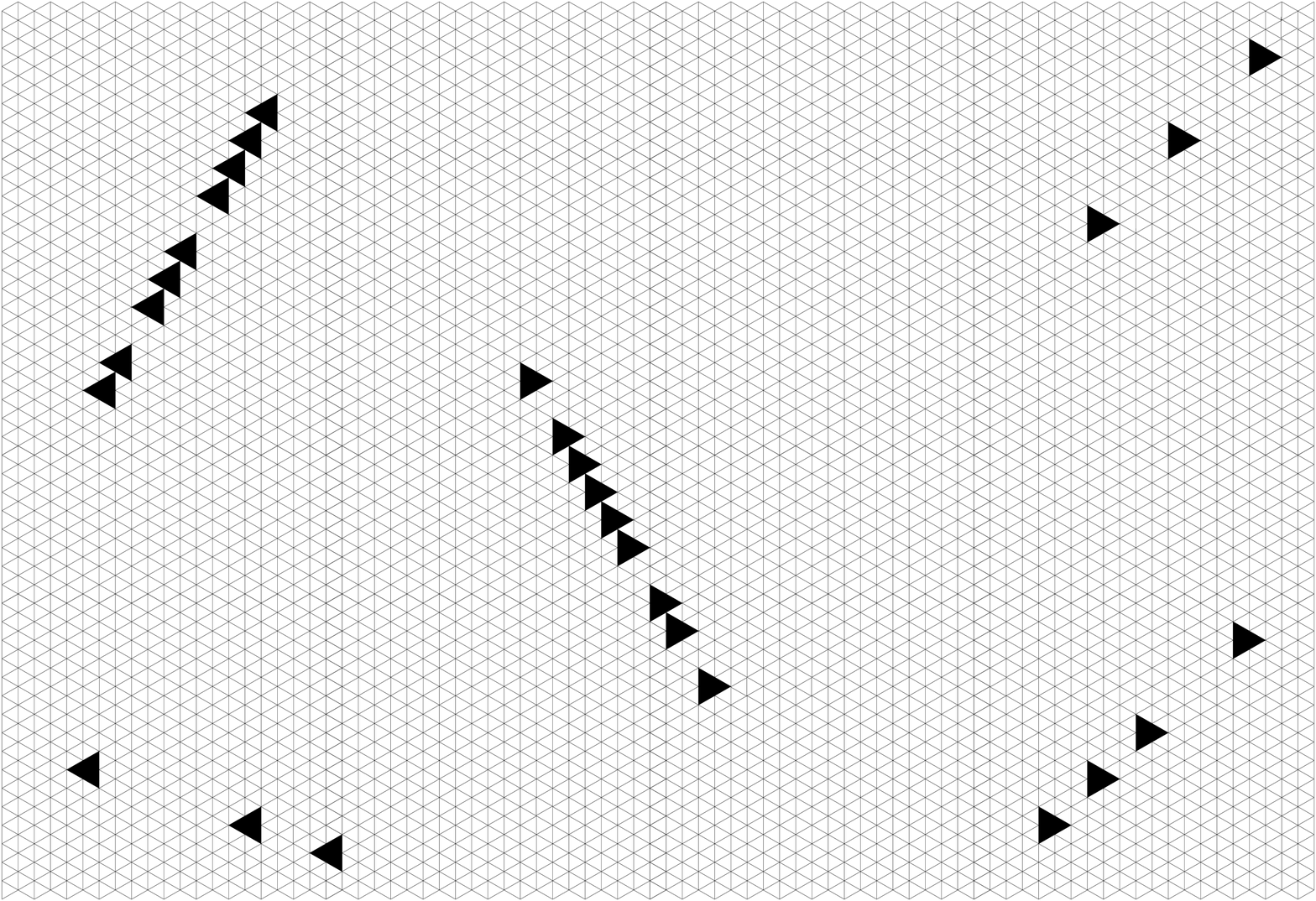}}
\hfill
}
\caption{\label{fbb} Linear multiholes of various slopes.}
\end{figure}



Our results concern the behavior of the correlation of such linear multiholes in the limit of large separation distances between them. To phrase these results, we need some notation for translations of given holes. This requires us to distinguish between the {\it shape} of a hole $O$ --- two holes have the same shape if they are translates of one another --- and the {\it placement} $O(x,y)$ of a hole of a given shape, which is its translation with the property that some distinguished unit triangle in the shape (say the topmost and leftmost) acquires coordinates $(x,y)$. For brevity, we refer to a hole placed at a specific location simply as a hole.

The joint correlation $\hat\omega$ of holes is defined as follows (see \cite{ef}).
For $j=1,\dotsc,k$, let $Q_j$ be either a lozenge-hole or a lattice triangular 
hole of side two. Define the {\it charge} $\q(O)$ of the hole $O$ to be the number of right-pointing unit triangles in $O$ minus the number of left-pointing unit triangles in $O$.
 
It is enough to define $\hat\omega(Q_1,\dotsc,Q_k)$ when the total charge $c=\sum_{j=1}^k \ch(Q_j)\geq0$ (the other case reduces 
to this by reflection across a vertical lattice line). Our definition is inductive on $c$:

\medskip
$(i)$. If $c=0$, let $N$ be large enough so that the lattice rhombus of side $N$ centered at the origin encloses 
all $Q_j$'s, and denote by $T_N$ the torus obtained from this large lattice rhombus by identifying its opposite 
sides. Set\footnote{\,$\M(\mathcal R)$ denotes the number of lozenge tilings of the lattice region $\mathcal R$.}
\begin{equation*}
\hat\omega(Q_1,\dotsc,Q_k):=
\lim_{N\to\infty}\frac{\M\left(T_N\setminus Q_1\cup\cdots \cup Q_k\right)}{\M\left(T_N\right)}.
\end{equation*}
$(ii)$. If $c>0$, define
\begin{equation*}
\hat\omega(Q_1,\dotsc,Q_k):=\lim_{R\to\infty}R^c\,\hat\omega\left(Q_1,\dotsc,Q_k,{\begin{matrix} & \hskip-0.153in \triangleleft \\[-10.45pt]
\triangleleft & \hskip-0.153in\triangleright \\[-10.45pt]
& \hskip-0.153in\triangleleft \end{matrix}}_{R,0}\right).
\end{equation*}

The above limits exist by Proposition \ref{tca}.

Given a hole $O$, we denote by $O+(x,y)$ its translation by the vector $(x,y)$.


For a given collection of holes $O_1,\dotsc,O_n$, 
define the {\it relative change $T^{O_1}_{\al,\be}(O_1,\dotsc,O_n)$ in the correlation $\omega(O_1,\dotsc,O_n)$ under displacing $O_1$ by $(\al,\be)$} by
\begin{equation}
T^{O_1}_{\al,\be}(O_1,\dotsc,O_n):=
\frac
{\omega(O_1+(\al,\be),\dotsc,O_n)}
{\omega(O_1,\dotsc,O_n)}
-1.
\label{eba}
\end{equation}

Define
\begin{equation}
{\bold E}(a,b;c,d):=\dfrac{(c-a,d-b)}{(c-a)^2+(c-a)(d-b)+(d-b)^2},
\label{ebaa}
\end{equation}
the vector pointing from $(a,b)$ to $(c,d)$ (two marked points in our $60^\circ$ coordinate system) and having length 
$\dfrac{1}{\dee((a,b),(c,d))}$ (where $\dee$ is the Euclidean distance).

%

\bigskip
\begin{theo}
\label{tba}
Suppose $O_i$ is either of type
${\begin{matrix} \triangleright &  \\[-10.45pt] \triangleleft & \hskip-0.233in\triangleright \\[-10.45pt] \triangleright & \end{matrix}}_{\!\!\!\!\!{\bold a}_i}^{\!\!\!\!\!q_i}$
with $3|1-q_i$, or of type
${\begin{matrix} & \hskip-0.153in \triangleleft \\[-10.45pt]
\triangleleft & \hskip-0.153in\triangleright \\[-10.45pt]
& \hskip-0.153in\triangleleft \end{matrix}}_{\,{\bold b}_i}^{\,q'_i}$
with $3|1-q'_i$, for $i=1,\dotsc,n$.
Let $x_1^{(R)},\dotsc,x_n^{(R)},y_1^{(R)},\dotsc,y_n^{(R)}\in\Z$ with
$\lim_{R\to\infty}{x_i^{(R)}}/{R}=x_i$, $\lim_{R\to\infty}{y_i^{(R)}}/{R}=y_i$, $(x_i,y_i)$ distinct.

Then if $O_i^{(R)}=O_i(x_i^{(R)},y_i^{(R)})$, for any fixed pair of integers $(\al,\be)\neq(0,0)$ we have as $R\to\infty$ that
\begin{align}
\frac{1}{\q(O_1)}\frac{1}{\sqrt{\alpha^2+\alpha\beta+\beta^2}}
\,T_{\alpha,\beta}^{O_1^{(R)}}(O_1^{(R)},\dotsc,O_n^{(R)})
=
\proj_{(\alpha,\beta)}\frac12\sum_{j=2}^n \q(O_j){\bold E(x_j,y_j;x_1,y_1)}\frac{1}{R}+o\left(\frac{1}{R}\right),
\label{ebb}
\end{align}
where $\proj_{(a,b)}(x,y)$ is the length of the orthogonal projection of vector $(x,y)$ on the direction of the vector $(a,b)$. 
\end{theo}

\parindent0pt
{\it Remark $1$}. Setting
\begin{equation}
{\bold T}=\frac12\sum_{j=2}^n \q(O_j){\bold E(x_j,y_j;x_1,y_1)}
\label{ebbb}
\end{equation}  
(if we consider a 2D system of electrical charges obtained by replacing $O_i$ by a charge of magnitude $\q(O_i)$, $i=2,\dotsc,n$, then ${\bold T}$ is --- up to a multiplicative constant --- the electric field produced by this system at the position of $O_1$),
equation \eqref{ebb} states that 
\begin{equation}
\,T_{\alpha,\beta}^{O_1^{(R)}}(O_1^{(R)},\dotsc,O_n^{(R)})
\sim
{\bold T}\cdot(\al,\be)\frac{\q(O_1)}{R},\ \ \ R\to\infty.
\label{ebbbb}
\end{equation}

{\it Remark $2$}. In the special case when $x_1^{(R)}=Rx_1,\dotsc,w_n^{(R)}=Rw_n$, the error term in \eqref{ebb} is $O(1/R^2)$. This follows from Propositions 4.1 and 4.5 of \cite{ef} by following through the steps of our proof of Theorem~\ref{tba}.

\parindent15pt

\section{Reducing to exact determinant evaluations}

The starting point for proving our results is a determinant formula which gives the exact value of the correlation of any collection $\mathcal{C}$ of unit holes on the triangular lattice with the property that $\mathcal{C}$ can be partitioned into pairs of adjacent unit triangles (i.e. unit triangles sharing at least one vertex).
The case when the numbers of left- and right-pointing unit holes are equal is due to Kenyon \cite{KeLocal}; we worked out the general case in  \cite{ef}. 

This formula involves the coupling function $P(x,y)$, $x,y\in\Z$ specified by
\begin{equation}
P(x,y)=\frac{1}{2\pi i}\int_{e^{2\pi i/3}}^{e^{4\pi i/3}} t^{-y-1}(-1-t)^{-x-1}dt,\ \ \ x\leq-1
\label{eca}
\end{equation}
and the symmetries $P(x,y)=P(y,x)=P(-x-y-1,x)$ (see \cite{KeLocal}),
and the coefficients $U_s$ of its asymptotic series
\begin{equation}
P(-3r-1+a,-1+b)\sim\sum_{s=0}^\infty (3r)^{-s-1}U_s(a,b),\ \ \ r\to\infty,\ a,b\in\Z
\label{ecb}
\end{equation}
which we found explicitly in \cite[Proposition 4.1]{ec} to be
\begin{equation}
U_s(a,b)=-\frac{i}{2\pi}\left.\left[\zeta^{a-b-1}(1-D\zeta^{-1})^{-b}-
         \zeta^{-a+b+1}(1-D\zeta)^{-b}\right](x^s)\right|_{x=a+b-1},
\label{ecbb}
\end{equation}
where $D$ is the difference operator\footnote{ For a function $f$ defined on $\Z$, $Df$
is defined by $Df(x)=f(x+1)-f(x)$.} and $\zeta=e^{2\pi i/3}$.

Denote by $r(a,b)$ and $l(a,b)$ the right- and left-pointing monomers (i.e. unit triangles) of coordinates $(a,b)$, 
respectively. The above mentioned formula from \cite{ef} is the following\footnote{ The matrix $M_U$ used in \cite[Proposition 2.1]{ef} is obtained from the one shown here by swapping columns $2i-1$ and $2i$, for $i=1,\cdots(m-n)/2-1$; this clearly doesn't change the absolute value of the determinant on the right hand side of \eqref{ecc}.}.

\begin{prop}[\cite{ef}]
\label{tca}
  Assume that $\{r(a_1,b_1),\dotsc,r(a_m,b_m),l(c_1,d_1),\dotsc,l(c_n,d_n)\}$
can be partitioned into subsets of size two so that the monomers in each subset share at least one 
vertex.
Then if $m\geq n$ we have
\begin{equation}
\hat{\omega}(r(a_1,b_1),\dotsc,r(a_m,b_m),l(c_1,d_1),\dotsc,l(c_n,d_n))=
\left|\det\left[\begin{matrix} M_P\!\!\!\!\!\!\!\!&&M_U\end{matrix}\right]\right|,
\label{ecc}
\end{equation}
where
\begin{equation*}
M_P=\left[
\begin{matrix}
P(a_1-c_1,b_1-d_1)&&\cdots&&P(a_1-c_n,b_1-d_n)\\
P(a_2-c_1,b_2-d_1)&&\cdots&&P(a_2-c_n,b_2-d_n)\\
\vdots&&\ &&\vdots\\
P(a_m-c_1,b_m-d_1)&&\cdots&&P(a_m-c_n,b_m-d_n)
\end{matrix}\right]
\end{equation*}
and 
\begin{equation*}
M_U=
\left[
\begin{matrix}
U_0(a_1+1,b_1)\!\!\!\!\!\!\!\!&&U_0(a_1,b_1+1)\!\!\!\!&&\cdots\!\!\!\!&&
U_{\frac{m-n}{2}-1}(a_1+1,b_1)\!\!\!\!\!\!\!\!&&U_{\frac{m-n}{2}-1}(a_1,b_1+1)\\
U_0(a_2+1,b_2)\!\!\!\!\!\!\!\!&&U_0(a_2,b_2+1)\!\!\!\!&&\cdots\!\!\!\!&&
U_{\frac{m-n}{2}-1}(a_2+1,b_2)\!\!\!\!\!\!\!\!&&U_{\frac{m-n}{2}-1}(a_2,b_2+1)\\
\vdots&&\vdots&&\ &&\vdots&&\vdots\\
U_0(a_m+1,b_m)\!\!\!\!\!\!\!\!&&U_0(a_m,b_m+1)\!\!\!\!&&\cdots\!\!\!\!&&
U_{\frac{m-n}{2}-1}(a_m+1,b_m)\!\!\!\!\!\!\!\!&&U_{\frac{m-n}{2}-1}(a_m,b_m+1)
\end{matrix}\right]\!\!.
\end{equation*}

\end{prop}

For lists of integers ${\bold a}=(a_1,\dotsc,a_s)$, ${\bold b}=(b_1,\dotsc,b_t)$ and rational numbers $q$, $q'$ with the property that $qa_i,q'b_j\in\Z$, $i=1,\dotsc,s$, $j=1,\dotsc,t$, define the $2s\times 2t$ matrix $A_{x,y,z,w}({\bold a},{\bold b};q,q')$ by

\begin{align}
&
  A_{x,y,z,w}({\bold a},{\bold b};q,q'):=
  \nonumber
\\[15pt]
&\ \ \ \ \ \ \ 
\left[
\begin{matrix}
P\left({\!\!\!\!\!\!\!\!\!\!\!x-z+a_1-b_1-1,\atop\ y-w+qa_1-q'b_1-1\!}\right)
&
P\left({\!\!\!\!x-z+a_1-b_1-2,\atop\ y-w+qa_1-q'b_1\!}\right)
&
\cdots
&
P\left({\!\!\!\!\!\!\!\!\!\!\!x-z+a_1-b_t-1,\atop\ y-w+qa_1-q'b_t-1\!}\right)
&
P\left({\!\!\!\!x-z+a_1-b_t-2,\atop\ y-w+qa_1-q'b_t\!}\right)
\\[10pt]
P\left({\!\!\!\!\!\!\!\!\!\!\!\!\!\!\!\!\!\!x-z+a_1-b_1,\atop\ y-w+qa_1-q'b_1-2\!}\right)
&
P\left({\!\!\!\!\!\!\!\!\!\!\!x-z+a_1-b_1-1,\atop\ y-w+qa_1-q'b_1-1\!}\right)
&
\cdots
&
P\left({\!\!\!\!\!\!\!\!\!\!\!\!\!\!\!\!\!\!x-z+a_1-b_t,\atop\ y-w+qa_1-q'b_t-2\!}\right)
&
P\left({\!\!\!\!\!\!\!\!\!\!\!x-z+a_1-b_t-1,\atop\ y-w+qa_1-q'b_t-1\!}\right)
\\
\\
P\left({\!\!\!\!\!\!\!\!\!\!\!x-z+a_2-b_1-1,\atop\ y-w+qa_2-q'b_1-1\!}\right)
&
P\left({\!\!\!\!x-z+a_2-b_1-2,\atop\ y-w+qa_2-q'b_1\!}\right)
&
\cdots
&
P\left({\!\!\!\!\!\!\!\!\!\!\!x-z+a_2-b_t-1,\atop\ y-w+qa_2-q'b_t-1\!}\right)
&
P\left({\!\!\!\!x-z+a_2-b_t-2,\atop\ y-w+qa_2-q'b_t\!}\right)
\\[10pt]
P\left({\!\!\!\!\!\!\!\!\!\!\!\!\!\!\!\!\!\!x-z+a_2-b_1,\atop\ y-w+qa_2-q'b_1-2\!}\right)
&
P\left({\!\!\!\!\!\!\!\!\!\!\!x-z+a_2-b_1-1,\atop\ y-w+qa_2-q'b_1-1\!}\right)
&
\cdots
&
P\left({\!\!\!\!\!\!\!\!\!\!\!\!\!\!\!\!\!\!x-z+a_2-b_t,\atop\ y-w+qa_2-q'b_t-2\!}\right)
&
P\left({\!\!\!\!\!\!\!\!\!\!\!x-z+a_2-b_t-1,\atop\ y-w+qa_2-q'b_t-1\!}\right)
\\
\\
\cdot & \cdot & & \cdot & \cdot\\
\cdot & \cdot & & \cdot & \cdot\\
\cdot & \cdot & & \cdot & \cdot
\\
\\
P\left({\!\!\!\!\!\!\!\!\!\!\!x-z+a_s-b_1-1,\atop\ y-w+qa_s-q'b_1-1\!}\right)
&
P\left({\!\!\!\!x-z+a_s-b_1-2,\atop\ y-w+qa_s-q'b_1\!}\right)
&
\cdots
&
P\left({\!\!\!\!\!\!\!\!\!\!\!x-z+a_s-b_t-1,\atop\ y-w+qa_s-q'b_t-1\!}\right)
&
P\left({\!\!\!\!x-z+a_s-b_t-2,\atop\ y-w+qa_s-q'b_t\!}\right)
\\[10pt]
P\left({\!\!\!\!\!\!\!\!\!\!\!\!\!\!\!\!\!\!x-z+a_s-b_1,\atop\ y-w+qa_s-q'b_1-2\!}\right)
&
P\left({\!\!\!\!\!\!\!\!\!\!\!x-z+a_s-b_1-1,\atop\ y-w+qa_s-q'b_1-1\!}\right)
&
\cdots
&
P\left({\!\!\!\!\!\!\!\!\!\!\!\!\!\!\!\!\!\!x-z+a_s-b_t,\atop\ y-w+qa_s-q'b_t-2\!}\right)
&
P\left({\!\!\!\!\!\!\!\!\!\!\!x-z+a_s-b_t-1,\atop\ y-w+qa_s-q'b_t-1\!}\right)
\end{matrix}
\right].
\nonumber
\\[5pt]
\label{ecd}
\end{align}
For $k\geq 0$, define also the $2s\times 2k $ matrix $B_{x,y}({\bold a},k;q,q')$ by

\begin{align}
  B_{x,y}({\bold a},k;q):=
  \nonumber
\left[
\begin{matrix}
U_0\left({\!\!\!\!\!x+a_1,\atop\ y+qa_1\!}\right)
&
U_0\left({\!\!\!\!\!x+a_1-1,\atop\ y+qa_1+1\!}\right)
&
\cdots
&
U_{k-1}\left({\!\!\!\!\!x+a_1,\atop\ y+qa_1\!}\right)
&
U_{k-1}\left({\!\!\!\!\!x+a_1-1,\atop\ y+qa_1+1\!}\right)
\\[10pt]
U_0\left({\!\!\!\!\!x+a_1+1,\atop\ y+qa_1-1\!}\right)
&
U_0\left({\!\!\!\!\!x+a_1,\atop\ y+qa_1\!}\right)
&
\cdots
&
U_{k-1}\left({\!\!\!\!\!x+a_1+1,\atop\ y+qa_1-1\!}\right)
&
U_{k-1}\left({\!\!\!\!\!x+a_1,\atop\ y+qa_1\!}\right)
\\
\\
U_0\left({\!\!\!\!\!x+a_2,\atop\ y+qa_2\!}\right)
&
U_0\left({\!\!\!\!\!x+a_2-1,\atop\ y+qa_2+1\!}\right)
&
\cdots
&
U_{k-1}\left({\!\!\!\!\!x+a_2,\atop\ y+qa_2\!}\right)
&
U_{k-1}\left({\!\!\!\!\!x+a_2-1,\atop\ y+qa_2+1\!}\right)
\\[10pt]
U_0\left({\!\!\!\!\!x+a_2+1,\atop\ y+qa_2-1\!}\right)
&
U_0\left({\!\!\!\!\!x+a_2,\atop\ y+qa_2\!}\right)
&
\cdots
&
U_{k-1}\left({\!\!\!\!\!x+a_2+1,\atop\ y+qa_2-1\!}\right)
&
U_{k-1}\left({\!\!\!\!\!x+a_2,\atop\ y+qa_2\!}\right)
\\
\\
\cdot & \cdot & & \cdot & \cdot\\
\cdot & \cdot & & \cdot & \cdot\\
\cdot & \cdot & & \cdot & \cdot
\\
\\
U_0\left({\!\!\!\!\!x+a_s,\atop\ y+qa_s\!}\right)
&
U_0\left({\!\!\!\!\!x+a_s-1,\atop\ y+qa_s+1\!}\right)
&
\cdots
&
U_{k-1}\left({\!\!\!\!\!x+a_s,\atop\ y+qa_s\!}\right)
&
U_{k-1}\left({\!\!\!\!\!x+a_s-1,\atop\ y+qa_s+1\!}\right)
\\[10pt]
U_0\left({\!\!\!\!\!x+a_s+1,\atop\ y+qa_s-1\!}\right)
&
U_0\left({\!\!\!\!\!x+a_s,\atop\ y+qa_s\!}\right)
&
\cdots
&
U_{k-1}\left({\!\!\!\!\!x+a_s+1,\atop\ y+qa_s-1\!}\right)
&
U_{k-1}\left({\!\!\!\!\!x+a_s,\atop\ y+qa_s\!}\right)
\end{matrix}
\right].
\nonumber
\\[5pt]
\label{ece}
\end{align}

Notice that if only the two unit holes lined up along the vertical side of a triangle of side two are present, the lozenge that fits in the notch between them is forced to be part of every tiling of the complement of these unit holes. Therefore, from the point of view of computing the correlation, such a pair of unit holes is equivalent to the triangular hole of side-length two that encloses them.

Suppose the list ${\bold a}_i$ has $s_i$ elements, $i=1,\dotsc,m$, and that ${\bold b}_j$ has $t_j$ elements, $j=1,\dotsc,n$. Using the observation in the previous paragraph, one readily sees that if $\sum_{i=1}^m s_i\geq\sum_{j=1}^n t_j$ (an assumption we can make without loss of generality), Proposition \ref{tca} implies that
\begin{align}
&\!\!\!\!\!\!\!\!\!\!\!\!\!\!\!\!
\omega\left({\begin{matrix}
\triangleright &  \\[-10.45pt] \triangleleft & \hskip-0.233in\triangleright \\[-10.45pt] \triangleright & \end{matrix}}_{\!\!\!\!\!\bold a_1}^{\!\!\!\!\!q_1}(x_1^{(R)},y_1^{(R)}),
\dotsc,
{\begin{matrix} \triangleright &  \\[-10.45pt] \triangleleft & \hskip-0.233in\triangleright \\[-10.45pt] \triangleright & \end{matrix}}_{\!\!\!\!\!\bold a_m}^{\!\!\!\!\!q_m}(x_m^{(R)},y_m^{(R)}),
{\begin{matrix} & \hskip-0.153in \triangleleft \\[-10.45pt]
\triangleleft & \hskip-0.153in\triangleright \\[-10.45pt]
& \hskip-0.153in\triangleleft \end{matrix}}_{\,\bold b_1}^{\,q'_1}(z_1^{(R)},w_1^{(R)}),
\dotsc,
{\begin{matrix} & \hskip-0.153in \triangleleft \\[-10.45pt]
\triangleleft & \hskip-0.153in\triangleright \\[-10.45pt]
& \hskip-0.153in\triangleleft \end{matrix}}_{\,\bold b_n}^{\,q'_n}(z_n^{(R)},w_n^{(R)})
\right)
\nonumber
\\[10pt]
&\ \ \ \ \ \ \ \ \
=
\omega\left(
\begin{matrix} \triangleright  \\[-7.85pt] \triangleright \end{matrix}_{\,\bold a_1}^{\,q_1}(x_1^{(R)},y_1^{(R)}),
\dotsc,
\begin{matrix} \triangleright  \\[-7.85pt] \triangleright \end{matrix}_{\,\bold a_m}^{\,q_m}(x_m^{(R)},y_m^{(R)}),
\begin{matrix} \triangleleft  \\[-7.85pt] \triangleleft \end{matrix}_{\,\bold b_1}^{\,q'_1}(z_1^{(R)},w_1^{(R)}),
\dotsc,
\begin{matrix} \triangleleft  \\[-7.85pt] \triangleleft \end{matrix}_{\,\bold b_n}^{\,q'_n}(z_n^{(R)},w_n^{(R)})
\right)
\nonumber
\\[10pt]
&\ \ \ \ \ \ \ \ \ \ 
=|\det(M)|,
\label{ecf}
\end{align}
with
\begin{align}
&
M=
\nonumber
\\[10pt]
&\ \ \ 
\left[
\begin{matrix}
A_{x_1^{(R)},y_1^{(R)},z_1^{(R)},w_1^{(R)}}({\bold a}_1,{\bold b}_1;q_1,q'_1) &
\dotsc&
A_{x_1^{(R)},y_1^{(R)},z_n^{(R)},w_n^{(R)}}({\bold a}_1,{\bold b}_n;q_1,q'_n) &
B_{x_1^{(R)},y_1^{(R)}}({\bold a}_1,\nu;q_1)
\\
\\
A_{x_2^{(R)},y_2^{(R)},z_1^{(R)},w_1^{(R)}}({\bold a}_2,{\bold b}_1;q_2,q'_1) &
\dotsc&
A_{x_2^{(R)},y_2^{(R)},z_n^{(R)},w_n^{(R)}}({\bold a}_2,{\bold b}_n;q_2,q'_n) &
B_{x_2^{(R)},y_2^{(R)}}({\bold a}_2,\nu;q_2)
\\
\\
\cdot & & \cdot & \cdot\\
\cdot & & \cdot & \cdot\\
\cdot & & \cdot & \cdot\\
\\
\\
A_{x_m^{(R)},y_m^{(R)},z_1^{(R)},w_1^{(R)}}({\bold a}_m,{\bold b}_1;q_m,q'_1) &
\dotsc&
A_{x_m^{(R)},y_m^{(R)},z_n^{(R)},w_n^{(R)}}({\bold a}_m,{\bold b}_n;q_2,q'_n) &
B_{x_m^{(R)},y_m^{(R)}}({\bold a}_m,\nu;q_m)
\end{matrix}  
\right],
\nonumber
\\
\label{ecg}
\end{align}
where $\nu=\sum_{i=1}^m s_i-\sum_{j=1}^n t_j$.

\medskip
Replacing $x_1^{(R)}$ by $x_1^{(R)}+\alpha$ and $y_1^{(R)}$ by $y_1^{(R)}+\beta$, \eqref{ecf} becomes
\begin{align}
&
\omega\left({\begin{matrix}
\triangleright &  \\[-10.45pt] \triangleleft & \hskip-0.233in\triangleright \\[-10.45pt] \triangleright & \end{matrix}}_{\!\!\!\!\!\bold a_1}^{\!\!\!\!\!q_1}(x_1^{(R)}+\alpha,y_1^{(R)}+\beta),
\dotsc,
{\begin{matrix} \triangleright &  \\[-10.45pt] \triangleleft & \hskip-0.233in\triangleright \\[-10.45pt] \triangleright & \end{matrix}}_{\!\!\!\!\!\bold a_m}^{\!\!\!\!\!q_m}(x_m^{(R)},y_m^{(R)}),
{\begin{matrix} & \hskip-0.153in \triangleleft \\[-10.45pt]
\triangleleft & \hskip-0.153in\triangleright \\[-10.45pt]
& \hskip-0.153in\triangleleft \end{matrix}}_{\,\bold b_1}^{\,q'_1}(z_1^{(R)},w_1^{(R)}),
\dotsc,
{\begin{matrix} & \hskip-0.153in \triangleleft \\[-10.45pt]
\triangleleft & \hskip-0.153in\triangleright \\[-10.45pt]
& \hskip-0.153in\triangleleft \end{matrix}}_{\,\bold b_n}^{\,q'_n}(z_n^{(R)},w_n^{(R)})
\right)
\nonumber
\\[10pt]
&\ \ \ \ \ \ \ \ \ \ 
=|\det(M_{\alpha,\beta})|
\label{ech}
\end{align}
where $M_{\alpha,\beta}$ is the matrix
\begin{align}
\left[
\begin{matrix}
A_{x_1^{(R)}+\alpha,y_1^{(R)}+\beta,z_1^{(R)},w_1^{(R)}}({\bold a}_1,{\bold b}_1;q_1,q'_1) &
\dotsc&
A_{x_1^{(R)}+\alpha,y_1^{(R)}+\beta,z_n^{(R)},w_n^{(R)}}({\bold a}_1,{\bold b}_n;q_1,q'_n) &
B_{x_1^{(R)}+\alpha,y_1^{(R)}+\beta}({\bold a}_1,\nu;q_1)
\\
\\
A_{x_2^{(R)},y_2^{(R)},z_1^{(R)},w_1^{(R)}}({\bold a}_2,{\bold b}_1;q_2,q'_1) &
\dotsc&
A_{x_2^{(R)},y_2^{(R)},z_n^{(R)},w_n^{(R)}}({\bold a}_2,{\bold b}_n;q_2,q'_n) &
B_{x_2^{(R)},y_2^{(R)}}({\bold a}_2,\nu;q_2)
\\
\\
\cdot & & \cdot & \cdot\\
\cdot & & \cdot & \cdot\\
\cdot & & \cdot & \cdot\\
\\
\\
A_{x_m^{(R)},y_m^{(R)},z_1^{(R)},w_1^{(R)}}({\bold a}_m,{\bold b}_1;q_m,q'_1) &
\dotsc&
A_{x_m^{(R)},y_m^{(R)},z_n^{(R)},w_n^{(R)}}({\bold a}_m,{\bold b}_n;q_2,q'_n) &
B_{x_m^{(R)},y_m^{(R)}}({\bold a}_m,\nu;q_m)
\end{matrix}  
\right]
\nonumber
\\
\label{eci}
\end{align}
(in particular, the matrix $M$ given by \eqref{ecg} is just $M_{0,0}$).

In order to prove Theorem \ref{tba}, we need to determine the $R\to\infty$ asymptotics of
\begin{equation}
\left|\frac{\det M_{\alpha,\beta}}{\det M}\right|-1. 
\label{ecj}
\end{equation}

The first obstacle to surmount is that the matrices in \eqref{ecj} are in general {\it asymptotically singular} --- i.e. the matrices formed by the dominant terms of their entries are singular (more precisely, as we will see below, this happens unless all $s_i$'s and $t_j$'s are equal to 1). We resolve this by repeatedly applying a determinant preserving operator (which we first considered in \cite[Section 3]{ec}) --- patterned on Newton's divided differences operator --- that acts on certain groups of rows and columns of the matrices $M$ and $M_{\alpha,\beta}$, and transforms them into asymptotically non-singular matrices $M'$ and $M'_{\alpha,\beta}$.

Namely, consider the following operation on a square matrix $X$ in which rows 
$i_1,\dotsc,i_k$ are of the form
$f(c_1),\dotsc,f(c_k)$, respectively, for some vector function $f$: Transform
rows $i_1,\dotsc,i_k$ of $X$ as
\begin{equation}
\left[\begin{matrix}
f(c_1)\\
f(c_2)\\
.\\
.\\
.\\
f(c_k)
\end{matrix}\right]
\mapsto
\left[\begin{matrix}
{\mathcal D}^0f(c_1)\\
(c_2-c_1){\mathcal D}^1 f(c_1)\\
.\\
.\\
.\\
(c_k-c_1)(c_k-c_2)\dotsc(c_k-c_{k-1}){\mathcal D}^{k-1}f(c_1)
\end{matrix}\right],
\label{eck}
\end{equation}
where $\mathcal D$ is Newton's divided difference operator, whose powers are defined inductively
by ${\mathcal D}^0 f=f$ and 
${\mathcal D}^r f(c_j)=({\mathcal D}^{r-1} f(c_{j+1})- {\mathcal D}^{r-1} f(c_{j}))/(c_{j+r}-c_j)$, $r\geq1$ (see \cite{Jordan}) .

This operation has an obvious analog for columns.

As noted in \cite[Section 5]{ec} (and as can be seen by looking at \eqref{ecd}, \eqref{ece} and \eqref{ecg}),  
operation \eqref{eck} can be applied a total of $2m+2n$ different times to the matrix $M$: each row 
of block matrices in the expression for $M$ given by  \eqref{ecg}, \eqref{ecd} and \eqref{ece} provides two
opportunities 
(along the odd-indexed rows and along the even-indexed ones), and each column consisting of
$A$-blocks provides two more\footnote{ This is due to the structure of the linear multiholes, and to the fact that each multihole slope $q$ satisfies $3|1-q$. We need this so that for instance the asymptotics of the odd-index entries in each column of matrix \eqref{ecd} are evaluations of the same function $f$ (for such multiholes this follows from the special case $k=l=0$ of \eqref{ege}).}. Let $M'$ be the matrix obtained from $M$ after applying these
$2m+2n$ operations.

The matrix $M'$ can be described as follows\footnote{ See \cite[Section 5]{ec}, which details the special case $x_1^{(R)}=Rx_1,\dotsc,w_n^{(R)}=Rw_n$.}: While the $2\times 2$ matrix in position $(k,l)$ in the
$(i,j)$-block of the $P$-part of $M$ ($1\leq k\leq s_i$, $1\leq l\leq t_j$) is
\begin{equation}
\left[
\begin{matrix}
P\left({\!\!\!\!\!\!\!\!\!\!\!x_i^{(R)}-z_j^{(R)}+a_{ik}-b_{jl}-1,\atop\ y_i^{(R)}-w_j^{(R)}+q_ia_{ik}-q'_jb_{jl}-1\!}\right)
&
P\left({\!\!\!\!x_i^{(R)}-z_j^{(R)}+a_{ik}-b_{jl}-2,\atop\ y_i^{(R)}-w_j^{(R)}+q_ia_{ik}-q'_jb_{jl}\!}\right)
\\[10pt]
P\left({\!\!\!\!\!\!\!\!\!\!\!\!\!\!\!\!\!\!x_i^{(R)}-z_j^{(R)}+a_{ik}-b_{jl},\atop\ y_i^{(R)}-w_j^{(R)}+q_ia_{ik}-q'_jb_{jl}-2\!}\right)
&
P\left({\!\!\!\!\!\!\!\!\!\!\!x_i^{(R)}-z_j^{(R)}+a_{ik}-b_{jl}-1,\atop\ y_i^{(R)}-w_j^{(R)}+q_ia_{ik}-q'_jb_{jl}-1\!}\right)
\end{matrix}
\right],
\label{ecl}
\end{equation}
by our construction, the corresponding $2\times 2$ matrix in $M'$ is 
\begin{equation}
\alpha_k^{(i)}\beta_l^{(j)}{\mathcal D}^{l-1}_b{\mathcal D}^{k-1}_a
\left[
\begin{matrix}
P\left({\!\!\!\!\!\!\!\!\!\!\!x_i^{(R)}-z_j^{(R)}+a_{ik}-b_{jl}-1,\atop\ y_i^{(R)}-w_j^{(R)}+q_ia_{ik}-q'_jb_{jl}-1\!}\right)
&
P\left({\!\!\!\!x_i^{(R)}-z_j^{(R)}+a_{ik}-b_{jl}-2,\atop\ y_i^{(R)}-w_j^{(R)}+q_ia_{ik}-q'_jb_{jl}\!}\right)
\\[10pt]
P\left({\!\!\!\!\!\!\!\!\!\!\!\!\!\!\!\!\!\!x_i^{(R)}-z_j^{(R)}+a_{ik}-b_{jl},\atop\ y_i^{(R)}-w_j^{(R)}+q_ia_{ik}-q'_jb_{jl}-2\!}\right)
&
P\left({\!\!\!\!\!\!\!\!\!\!\!x_i^{(R)}-z_j^{(R)}+a_{ik}-b_{jl}-1,\atop\ y_i^{(R)}-w_j^{(R)}+q_ia_{ik}-q'_jb_{jl}-1\!}\right)
\end{matrix}
\right],
\label{ecm}
\end{equation}
where the powers of ${\mathcal D}$ act entry-wise (${\mathcal D}^{k-1}_a$ acting on the sequence $a_{i1},a_{i2},\dotsc,a_{is_i}$ consisting of the elements of ${\bold a}_i$, ${\mathcal D}^{l-1}_b$ on the sequence
$-b_{j1},-b_{j2},\dotsc,-b_{jt_j}$
consisting of
the negatives of
the elements of ${\bold b}_j$), and 
\begin{align}
\alpha_k^{(i)}&=(a_{ik}-a_{i1})(a_{ik}-a_{i2})\cdots(a_{ik}-a_{i,k-1})
\label{ecn}
\\
\beta_l^{(j)}&=(-b_{jl}+b_{j1})(-b_{jl}+b_{j2})\cdots(-b_{jl}+b_{j,l-1}).
\label{eco}
\end{align}
Similarly, while the $2\times 2$ sub-matrix of $M$ at the intersection of rows $2k-1$ and $2k$
of block-row $i$ with columns $2T+2l-1$ and $2T+2l$ ($1\leq k\leq s_i$, $1\leq l\leq S-T$) is 
\begin{equation}
\left[
\begin{matrix}
U_{l-1}\left({\!\!\!\!\!x_i^{(R)}+a_{ik},\atop\ y_i^{(R)}+q_ia_{ik}\!}\right)
&
U_{l-1}\left({\!\!\!\!\!x_i^{(R)}+a_{ik}-1,\atop\ y_i^{(R)}+q_ia_{ik}+1\!}\right)
\\[10pt]
U_{l-1}\left({\!\!\!\!\!x_i^{(R)}+a_{ik}+1,\atop\ y_i^{(R)}+q_ia_{ik}-1\!}\right)
&
U_{l-1}\left({\!\!\!\!\!x_i^{(R)}+a_{ik},\atop\ y_i^{(R)}+q_ia_{ik}\!}\right)
\end{matrix}\right],
\label{ecp}
\end{equation}
the corresponding $2\times 2$ sub-matrix of $M'$ is
\begin{equation}
\alpha_k^{(i)}{\mathcal D}^{k-1}_a
\left[
\begin{matrix}
U_{l-1}\left({\!\!\!\!\!x_i^{(R)}+a_{ik},\atop\ y_i^{(R)}+q_ia_{ik}\!}\right)
&
U_{l-1}\left({\!\!\!\!\!x_i^{(R)}+a_{ik}-1,\atop\ y_i^{(R)}+q_ia_{ik}+1\!}\right)
\\[10pt]
U_{l-1}\left({\!\!\!\!\!x_i^{(R)}+a_{ik}+1,\atop\ y_i^{(R)}+q_ia_{ik}-1\!}\right)
&
U_{l-1}\left({\!\!\!\!\!x_i^{(R)}+a_{ik},\atop\ y_i^{(R)}+q_ia_{ik}\!}\right)
\end{matrix}\right]
\label{ecq}
\end{equation}
with $\alpha_k^{(i)}$ given by \eqref{ecn}.

Since operations \eqref{eck} preserve the determinant (see \cite[Lemma 5.2]{ec}), we have
\begin{equation}
\det M=\det M',
\label{ecr}
\end{equation}
and as we will see below, $M'$ turns out to be asymptotically non-singular. Furthermore --- and this is crucial for our proof --- the determinant of the matrix consisting of the dominant terms of $M'$ turns out to have a simple product evaluation (see Theorem \ref{tdb}).

The matrix $M_{\alpha,\beta}$ at the numerator in \eqref{ecj} can be handled the same way. It differs from matrix $M$ only along the first block-row (consisting of rows $1,\dotsc,2s_1$). Applying to $M_{\alpha,\beta}$ precisely the same $2m+2n$ operations as to $M$, the resulting matrix $M'_{\alpha,\beta}$ is therefore given by \eqref{ecm}--\eqref{eco} and \eqref{ecq} for $i\geq 2$. For $i=1$, it follows from our construction that the 
$2\times 2$ matrix in  position $(k,l)$ in the
$(1,j)$-block of the $P$-part of $M'_{\alpha,\beta}$ ($1\leq k\leq s_1$, $1\leq l\leq t_j$) is 
\begin{equation}
\alpha_k^{(1)}\beta_l^{(j)}{\mathcal D}^{l-1}_b{\mathcal D}^{k-1}_a
\left[
\begin{matrix}
P\left({\!\!\!\!\!\!\!\!\!\!\!x_1^{(R)}-z_j^{(R)}+a_{1k}-b_{jl}+\alpha-1,\atop\ y_1^{(R)}-w_j^{(R)}+q_1a_{1k}-q'_jb_{jl}+\beta-1\!}\right)
&
P\left({\!\!\!\!x_1^{(R)}-z_j^{(R)}+a_{1k}-b_{jl}+\alpha-2,\atop\ y_1^{(R)}-w_j^{(R)}+q_1a_{1k}-q'_jb_{jl}+\beta\!}\right)
\\[10pt]
P\left({\!\!\!\!\!\!\!\!\!\!\!\!\!\!\!\!\!\!x_1^{(R)}-z_j^{(R)}+a_{1k}-b_{jl}+\alpha,\atop\ y_1^{(R)}-w_j^{(R)}+q_1a_{1k}-q'_jb_{jl}+\beta-2\!}\right)
&
P\left({\!\!\!\!\!\!\!\!\!\!\!x_1^{(R)}-z_j^{(R)}+a_{1k}-b_{jl}+\alpha-1,\atop\ y_1^{(R)}-w_j^{(R)}+q_1a_{1k}-q'_jb_{jl}+\beta-1\!}\right)
\end{matrix}
\right],
\label{ecs}
\end{equation}
and the $2\times 2$ sub-matrix of $M'_{\alpha,\beta}$ at the intersection of rows $2k-1$ and $2k$
of the first block-row  with columns $2T+2l-1$ and $2T+2l$ ($1\leq l\leq S-T$) is 
\begin{equation}
\alpha_k^{(1)}{\mathcal D}^{k-1}_a
\left[
\begin{matrix}
U_{l-1}\left({\!\!\!\!\!x_i^{(R)}+a_{ik}+\alpha,\atop\ y_i^{(R)}+q_ia_{ik}+\beta\!}\right)
&
U_{l-1}\left({\!\!\!\!\!x_i^{(R)}+a_{ik}+\alpha-1,\atop\ y_i^{(R)}+q_ia_{ik}+\beta+1\!}\right)
\\[10pt]
U_{l-1}\left({\!\!\!\!\!x_i^{(R)}+a_{ik}+\alpha+1,\atop\ y_i^{(R)}+q_ia_{ik}+\beta-1\!}\right)
&
U_{l-1}\left({\!\!\!\!\!x_i^{(R)}+a_{ik}+\alpha,\atop\ y_i^{(R)}+q_ia_{ik}+\beta\!}\right)
\end{matrix}\right],
\label{ect}
\end{equation}
where --- since the multihole
${\begin{matrix} \triangleright &  \\[-10.45pt] \triangleleft & \hskip-0.233in\triangleright \\[-10.45pt] \triangleright & \end{matrix}}_{\!\!\!\!\!{\bold a}_1}^{\!\!\!\!\!q_1}(x_1^{(R)}+\al,y_1^{(R)}+\be)$ has the same shape as
${\begin{matrix} \triangleright &  \\[-10.45pt] \triangleleft & \hskip-0.233in\triangleright \\[-10.45pt] \triangleright & \end{matrix}}_{\!\!\!\!\!{\bold a}_1}^{\!\!\!\!\!q_1}(x_1^{(R)},y_1^{(R)})$
--- the prefactors $\alpha_k^{(1)}$ and $\beta_l^{(j)}$ in \eqref{ecs} and \eqref{ect} are precisely the same as in \eqref{ecn} and \eqref{eco}.

%
\begin{prop}
\label{tcb}  
Let the matrices $\widetilde{M'}$ and $\widetilde{M'}_{\alpha,\beta}$ be defined precisely as the matrices $M'$ and $M'_{\alpha,\beta}$ given by  \eqref{ecm}--\eqref{eco}, \eqref{ecq} and \eqref{ecs}--\eqref{ect}, with the one exception that the prefactors $\alpha_k^{(i)}$ and $\beta_l^{(j)}$ are not included. Then 
\begin{equation}
\frac{\det M_{\alpha,\beta}}{\det M}
=
\frac{\det \widetilde{M'}_{\alpha,\beta}}{\det \widetilde{M'}}.
\label{ecu}
\end{equation}
\end{prop}

\begin{proof}
Since the operations that transformed $M$ and $M_{\alpha,\beta}$ into $M'$ and $M'_{\alpha,\beta}$ preserve the determinant, we have
\begin{equation*}
\frac{\det M_{\alpha,\beta}}{\det M}
=
\frac{\det {M'}_{\alpha,\beta}}{\det {M'}}.
\end{equation*}
Note that all the prefactors $\alpha_k^{(i)}$ and $\beta_l^{(j)}$ involved in the determinants on the right hand side above can be pulled out in front of the determinants, by suitably factoring them out along rows and columns. Since, as noted above,  the prefactors $\alpha_k^{(i)}$ and $\beta_l^{(j)}$ in \eqref{ecs} and \eqref{ect} are precisely the same as in \eqref{ecm} and \eqref{ecq}, all of them cancel out in the ratio on the right hand side above, yielding \eqref{ecu}.
\end{proof}  

\section{A product formula for $\det\lfloor\lfloor\widetilde{M}'\rfloor\rfloor$}

As we stated in the previous section, the determinant of the matrix $\lfloor\lfloor\widetilde{M}'\rfloor\rfloor$ formed by the dominant terms of the entries of $\widetilde{M}'$ turns out to be non-singular\footnote{ Throughout this paper, for a matrix $A$ whose entries depend on a large parameter, we denote by $\lfloor\lfloor A\rfloor\rfloor$ the matrix formed by the dominant parts (as the large parameter approaches infinity) of the entries of $A$.}. In fact, $\det\lfloor\lfloor\widetilde{M}'\rfloor\rfloor$ turns out to have an explicit, simple (and visibly non-zero) product expression.

By \eqref{ecm} and \eqref{ecq}, the matrix $\widetilde{M}'$ defined in the statement of Proposition \ref{tcb} is the block matrix
\begin{align}
&
\widetilde{M}'=
\nonumber
\\[10pt]
&
\left[
\begin{matrix}
A'_{x_1^{(R)},y_1^{(R)},z_1^{(R)},w_1^{(R)}}({\bold a}_1,{\bold b}_1;q_1,q'_1) &
\dotsc&
A'_{x_1^{(R)},y_1^{(R)},z_n^{(R)},w_n^{(R)}}({\bold a}_1,{\bold b}_n;q_1,q'_n) &
B'_{x_1^{(R)},y_1^{(R)}}({\bold a}_1,S-T;q_1)
\\
\\
A'_{x_2^{(R)},y_2^{(R)},z_1^{(R)},w_1^{(R)}}({\bold a}_2,{\bold b}_1;q_2,q'_1) &
\dotsc&
A'_{x_2^{(R)},y_2^{(R)},z_n^{(R)},w_n^{(R)}}({\bold a}_2,{\bold b}_n;q_2,q'_n) &
B'_{x_2^{(R)},y_2^{(R)}}({\bold a}_2,S-T;q_2)
\\
\\
\cdot & & \cdot & \cdot\\
\cdot & & \cdot & \cdot\\
\cdot & & \cdot & \cdot\\
\\
\\
A'_{x_m^{(R)},y_m^{(R)},z_1^{(R)},w_1^{(R)}}({\bold a}_m,{\bold b}_1;q_m,q'_1) &
\dotsc&
A'_{x_m^{(R)},y_m^{(R)},z_n^{(R)},w_n^{(R)}}({\bold a}_m,{\bold b}_n;q_2,q'_n) &
B'_{x_m^{(R)},y_m^{(R)}}({\bold a}_m,S-T;q_m)
\end{matrix}  
\right]
\nonumber
\\
\label{eda}
\end{align}
whose blocks are given by
\begin{align}
&
A'_{x,y,z,w}({\bold a},{\bold b};q,q')
=
\nonumber
\\[10pt]
&\ \ \ 
\left[
\begin{matrix}
& & \cdot & &\\  
& & \cdot & &\\  
& & \cdot & &\\  
 & {\mathcal D}^{l-1}_b {\mathcal D}^{k-1}_a P\left({\!\!\!\!\!\!\!\!\!\!\!x-z+a_{k}-b_{l}-1,\atop\ y-w+qa_{k}-q'b_{l}-1\!}\right)
&
&
{\mathcal D}^{l-1}_b {\mathcal D}^{k-1}_a P\left({\!\!\!\!x-z+a_{k}-b_{l}-2,\atop\ y-w+qa_{k}-q'b_{l}\!}\right)
&
\\
\ \ \ \cdots & & & & \cdots\ \ \ 
\\
&
{\mathcal D}^{l-1}_b {\mathcal D}^{k-1}_a P\left({\!\!\!\!\!\!\!\!\!\!\!\!\!\!\!\!\!\!x-z+a_{k}-b_{l},\atop\ y-w+qa_{k}-q'b_{l}-2\!}\right)
&
&
{\mathcal D}^{l-1}_b {\mathcal D}^{k-1}_a P\left({\!\!\!\!\!\!\!\!\!\!\!x-z+a_{k}-b_{l}-1,\atop\ y-w+qa_{k}-q'b_{l}-1\!}\right)
&
\\
& & \cdot & &\\  
& & \cdot & &\\  
& & \cdot & &\\  
\end{matrix}
\right]_{1\leq k\leq s,1\leq l\leq t}
\nonumber
\\
\label{tdb}
\end{align}
and
\begin{align}
B'_{x,y}({\bold a},S-T;q)
=
\left[
\begin{matrix}
& & \cdot & &\\  
& & \cdot & &\\  
& & \cdot & &\\  
 &{\mathcal D}^{k-1}_a  U_{l-1}\left({\!\!\!\!\!x+a_{k},\atop\ y+qa_{k}\!}\right)
&
&
{\mathcal D}^{k-1}_a U_{l-1}\left({\!\!\!\!\!x+a_{k}-1,\atop\ y+qa_{k}+1\!}\right)
&
\\
\ \ \ \cdots & & & & \cdots\ \ \ 
\\
&
{\mathcal D}^{k-1}_a U_{l-1}\left({\!\!\!\!\!x+a_{k}+1,\atop\ y+qa_{k}-1\!}\right)
&
&
{\mathcal D}^{k-1}_a U_{l-1}\left({\!\!\!\!\!x+a_{k},\atop\ y+qa_{k}\!}\right)
&
\\
& & \cdot & &\\  
& & \cdot & &\\  
& & \cdot & &\\  
\end{matrix}
\right]_{1\leq k\leq s,1\leq l\leq S-T}
\nonumber
\\
\label{edc}
\end{align}
%


Before giving the formula for the determinant of $\lfloor\lfloor \widetilde{M}'\rfloor\rfloor$, it will be useful to record its entries. They depend on the residue classes modulo 3 of the sequences $x_1^{(R)},\dotsc,w_n^{(R)}$.
The following result follows directly from equations \eqref{eda}--\eqref{edc}, \eqref{ege} and \eqref{egi}.

Throughout this paper $\zeta=e^{2\pi i/3}$.

\begin{prop}
\label{tda}
Assume $3|1-q_i,1-q'_j$, $i=1,\dotsc,m$, $j=1,\dotsc,n$.
Then if $x_1^{(R)}/R\to x_1,\dotsc,w_n^{(R)}/R\to w_n$ and $x_1^{(R)}=\alpha_1\,(\mod 3),\dotsc,w_n^{(R)}=\delta_n\,(\mod 3)$, we have
\begin{align}
&
\lfloor\lfloor\widetilde{M}'\rfloor\rfloor=\frac{1}{2\pi i}
\nonumber
\\[10pt]
&\ \ \ \ \ \ \ \ \ \ 
\times
\left[
\begin{matrix}
\dot{A}_{x_1,y_1,z_1,w_1}^{\alpha_1,\beta_1,\gamma_1,\delta_1}({\bold a}_1,{\bold b}_1;q_1,q'_1) &
\dotsc&
\dot{A}_{x_1,y_1,z_n,w_n}^{\alpha_1,\beta_1,\gamma_n,\delta_n}({\bold a}_1,{\bold b}_n;q_1,q'_n) &
\dot{B}_{x_1,y_1}^{\alpha_1,\beta_1}({\bold a}_1,S-T;q_1)
\\
\\
\dot{A}_{x_2,y_2,z_1,w_1}^{\alpha_2,\beta_2,\gamma_1,\delta_1}({\bold a}_2,{\bold b}_1;q_2,q'_1) &
\dotsc&
\dot{A}_{x_2,y_2,z_n,w_n}^{\alpha_2\beta_2,\gamma_n,\delta_n}({\bold a}_2,{\bold b}_n;q_2,q'_n) &
\dot{B}_{x_2,y_2}^{\alpha_2,\beta_2}({\bold a}_2,S-T;q_2)
\\
\\
\cdot & & \cdot & \cdot\\
\cdot & & \cdot & \cdot\\
\cdot & & \cdot & \cdot\\
\\
\\
\dot{A}_{x_m,y_m,z_1,w_1}^{\alpha_m,\beta_m,\gamma_1,\delta_1}({\bold a}_m,{\bold b}_1;q_m,q'_1) &
\dotsc&
\dot{A}_{x_m,y_m,z_n,w_n}^{\alpha_m,\beta_m,\gamma_n,\delta_n}({\bold a}_m,{\bold b}_n;q_2,q'_n) &
\dot{B}_{x_m,y_m}^{\alpha_m,\beta_m}({\bold a}_m,S-T;q_m)
\end{matrix}  
\right]
\nonumber
\\
\label{edd}
\end{align}
where the blocks $\dot{A}$ and $\dot{B}$ are given by\footnote{ Here and throughout the rest of the paper $\langle f(\zeta)\rangle$ stands for $f(\zeta)-f(\zeta^{-1})$.}
\begin{align}
&
\dot{A}_{x,y,z,w}^{\alpha,\beta,\gamma,\delta}({\bold a},{\bold b};q,q')
=
\nonumber
\\[10pt]
&
\left[
\begin{matrix}
& & \cdot & &\\  
& & \cdot & &\\  
& & \cdot & &\\  
 & \left\langle \frac{\zeta^{-1+\alpha-\beta+\gamma-\delta}{k+l-2\choose l-1}(1-q\zeta)^{k-1}(1-q'\zeta)^{l-1}}{[z-x-\zeta(w-y)]^{k+l-1}R^{k+l-1}}\right\rangle
&
&
\left\langle \frac{\zeta^{-3+\alpha-\beta+\gamma-\delta}{k+l-2\choose l-1}(1-q\zeta)^{k-1}(1-q'\zeta)^{l-1}}{[z-x-\zeta(w-y)]^{k+l-1}R^{k+l-1}}\right\rangle
&
\\
\cdots & & & & \cdots
\\
&
\left\langle \frac{\zeta^{1+\alpha-\beta+\gamma-\delta}{k+l-2\choose l-1}(1-q\zeta)^{k-1}(1-q'\zeta)^{l-1}}{[z-x-\zeta(w-y)]^{k+l-1}R^{k+l-1}}\right\rangle
&
&
\left\langle \frac{\zeta^{-1+\alpha-\beta+\gamma-\delta}{k+l-2\choose l-1}(1-q\zeta)^{k-1}(1-q'\zeta)^{l-1}}{[z-x-\zeta(w-y)]^{k+l-1}R^{k+l-1}}\right\rangle
&
\\
& & \cdot & &\\  
& & \cdot & &\\  
& & \cdot & &\\  
\end{matrix}
\right]_{1\leq k\leq s\atop 1\leq l\leq t}
\nonumber
\\[5pt]
\label{ede}
\end{align}
and
\begin{align}
\dot{B}_{x,y}^{\alpha,\beta}({\bold a},S-T;q)
=
\left[
\begin{matrix}
& & \cdot & &\\  
& & \cdot & &\\  
& & \cdot & &\\  
 &\left\langle \frac{\zeta^{-1+\alpha-\beta}{l-1\choose k-1}(1-q\zeta)^{k-1}}{(x-\zeta y)^{k-l}R^{k-l}}\right\rangle
&
&
\left\langle \frac{\zeta^{-3+\alpha-\beta}{l-1\choose k-1}(1-q\zeta)^{k-1}}{(x-\zeta y)^{k-l}R^{k-l}}\right\rangle
&
\\
\ \ \ \cdots & & & & \cdots\ \ \ 
\\
&
\left\langle \frac{\zeta^{1+\alpha-\beta}{l-1\choose k-1}(1-q\zeta)^{k-1}}{(x-\zeta y)^{k-l}R^{k-l}}\right\rangle
&
&
\left\langle \frac{\zeta^{-1+\alpha-\beta}{l-1\choose k-1}(1-q\zeta)^{k-1}}{(x-\zeta y)^{k-l}R^{k-l}}\right\rangle
&
\\
& & \cdot & &\\  
& & \cdot & &\\  
& & \cdot & &\\  
\end{matrix}
\right]_{1\leq k\leq s\atop 1\leq l\leq S-T}
\nonumber
\\
\label{edf}
\end{align}

\end{prop}

\begin{theo} Set $N=2\sum_{1\leq i<j\leq m}s_is_j+2\sum_{1\leq i<j\leq n}t_it_j-2\sum_{i=1}^m\sum_{j=1}^ns_it_j$. Then the determinant of the matrix \eqref{edd} in Proposition $\ref{tda}$ is
\label{tdb}
\begin{align}
\det \lfloor\lfloor \widetilde{M}'\rfloor\rfloor = 
\left(\frac{3}{4\pi^2}\right)^S
\prod_{1\leq i<j\leq m}&[(x_i-x_j)^2+(x_i-x_j)(y_i-y_j)+(y_i-y_j)^2]^{s_is_j}
\\
\times
\prod_{1\leq i<j\leq n}&[(z_i-z_j)^2+(z_i-z_j)(w_i-w_j)+(w_i-w_j)^2]^{t_it_j}
\\
\times
\prod_{i=1}^m\prod_{j=1}^n&[(x_i-z_j)^2+(x_i-z_j)(y_i-w_j)+(y_i-w_j)^2]^{-s_it_j}
R^N.
\label{tdd}
\end{align}
\end{theo}

\parindent0pt
{\it Remark $3$.} Note that, quite remarkably, even though the entries of $\lfloor\lfloor\widetilde{M'}\rfloor\rfloor$ depend on the residues $\alpha_1,\dotsc,\delta_n$ of $x_1^{(R)},\dotsc,w_n^{(R)}$ modulo 3, the value of its determinant does not!

\parindent15pt

\medskip
\begin{proof}
In the special case when $x_i^{(R)}=Rx_i$, $y_i^{(R)}=Ry_i$, $i=1,\dotsc,m$ and $z_j^{(R)}=Rz_j$, $w_j^{(R)}=Rw_j$, $j=1,\dotsc,n$, this is a direct consequence of Proposition 5.3 and Theorem 14.1 of \cite{ec}. The arguments presented there extend to the general case as explained in the proof of Proposition~3.1 of \cite{ef}, using Theorem 4.1, Proposition 4.5 and Lemma 3.3 of \cite{ef} (see also Remark 3.5 in \cite{ef}).
\end{proof}

\section{Proof of Theorem \ref{tba} when $3|\alpha-\beta$
        }






In our proof we will make use of the following elementary result.

\begin{lem}
\label{tea}   
Let $A$ and $B$ be two $n\times n$ matrices of the form
\begin{equation}
A:=\left[\begin{matrix}
a_{11} & \dots & a_{1n}\\
\cdot & & \cdot\\[-7pt]
\cdot & & \cdot\\[-7pt]
\cdot & & \cdot\\
a_{k1} & \dots & a_{kn}\\
x_{k+1,1} & \dots & x_{k+1,n}\\
\cdot & & \cdot\\[-7pt]
\cdot & & \cdot\\[-7pt]
\cdot & & \cdot\\
x_{n1} & \dots & x_{nn}
\end{matrix}\right],
\ \ 
B:=\left[\begin{matrix}
b_{11} & \dots & b_{1n}\\
\cdot & & \cdot\\[-7pt]
\cdot & & \cdot\\[-7pt]
\cdot & & \cdot\\
b_{k1} & \dots & b_{kn}\\
x_{k+1,1} & \dots & x_{k+1,n}\\
\cdot & & \cdot\\[-7pt]
\cdot & & \cdot\\[-7pt]
\cdot & & \cdot\\
x_{n1} & \dots & x_{nn}
\end{matrix}\right].
\label{eea}
\end{equation}
Then
\begin{equation}
\det(A)-\det(B)=\sum_{i=1}^k \det
\left[\begin{matrix}
b_{11} & \dots & b_{1n}\\
\cdot & & \cdot\\[-7pt]
\cdot & & \cdot\\[-7pt]
\cdot & & \cdot\\
b_{i-1,1} & \dots & b_{i-1,n}\\
a_{i1}-b_{i1} & \dots & a_{in}-b_{in}\\
a_{i+1,1} & \dots & a_{i+1,n}\\
\cdot & & \cdot\\[-7pt]
\cdot & & \cdot\\[-7pt]
\cdot & & \cdot\\
a_{k1} & \dots & a_{kn}\\
x_{k+1,1} & \dots & x_{k+1,n}\\
\cdot & & \cdot\\[-7pt]
\cdot & & \cdot\\[-7pt]
\cdot & & \cdot\\
x_{n1} & \dots & x_{nn}
\end{matrix}\right].
\label{eeb}
\end{equation}

\end{lem}

\begin{proof} Use the linearity of the determinant in row $i$ of the $i$th summand on the right hand side. In the telescoping cancellation that results, all but the terms on the left hand side cancel out. \end{proof}

%

\parindent0pt
\medskip
{\it Proof of Theorem $\ref{tba}$ for $3|\al-\be$.} Let
$
\left\{
{\begin{matrix}
\triangleright &  \\[-10.45pt] \triangleleft & \hskip-0.233in\triangleright \\[-10.45pt] \triangleright & \end{matrix}}_{\!\!\!\!\!\bold a_1}^{\!\!\!\!\!q_1},
\dotsc,
{\begin{matrix} \triangleright &  \\[-10.45pt] \triangleleft & \hskip-0.233in\triangleright \\[-10.45pt] \triangleright & \end{matrix}}_{\!\!\!\!\!\bold a_m}^{\!\!\!\!\!q_m},
{\begin{matrix} & \hskip-0.153in \triangleleft \\[-10.45pt]
\triangleleft & \hskip-0.153in\triangleright \\[-10.45pt]
& \hskip-0.153in\triangleleft \end{matrix}}_{\,\bold b_1}^{\,q'_1},
\dotsc,
{\begin{matrix} & \hskip-0.153in \triangleleft \\[-10.45pt]
\triangleleft & \hskip-0.153in\triangleright \\[-10.45pt]
& \hskip-0.153in\triangleleft \end{matrix}}_{\,\bold b_n}^{\,q'_n}
\right\}
$
be the given collection of holes.
By equations
\eqref{ecf}, \eqref{ech} and Proposition \ref{tcb} we have
%
\begin{align}
& 
\frac
{
\omega\left({\begin{matrix}
\triangleright &  \\[-10.45pt] \triangleleft & \hskip-0.233in\triangleright \\[-10.45pt] \triangleright & \end{matrix}}_{\!\!\!\!\!\bold a_1}^{\!\!\!\!\!q_1}(x_1^{(R)}+\alpha,y_1^{(R)}+\beta),
\dotsc,
{\begin{matrix} \triangleright &  \\[-10.45pt] \triangleleft & \hskip-0.233in\triangleright \\[-10.45pt] \triangleright & \end{matrix}}_{\!\!\!\!\!\bold a_m}^{\!\!\!\!\!q_m}(x_m^{(R)},y_m^{(R)}),
{\begin{matrix} & \hskip-0.153in \triangleleft \\[-10.45pt]
\triangleleft & \hskip-0.153in\triangleright \\[-10.45pt]
& \hskip-0.153in\triangleleft \end{matrix}}_{\,\bold b_1}^{\,q'_1}(z_1^{(R)},w_1^{(R)}),
\dotsc,
{\begin{matrix} & \hskip-0.153in \triangleleft \\[-10.45pt]
\triangleleft & \hskip-0.153in\triangleright \\[-10.45pt]
& \hskip-0.153in\triangleleft \end{matrix}}_{\,\bold b_n}^{\,q'_n}(z_n^{(R)},w_n^{(R)})
\right)
}
{
\omega\left({\begin{matrix}
\triangleright &  \\[-10.45pt] \triangleleft & \hskip-0.233in\triangleright \\[-10.45pt] \triangleright & \end{matrix}}_{\!\!\!\!\!\bold a_1}^{\!\!\!\!\!q_1}(x_1^{(R)},y_1^{(R)}),
\dotsc,
{\begin{matrix} \triangleright &  \\[-10.45pt] \triangleleft & \hskip-0.233in\triangleright \\[-10.45pt] \triangleright & \end{matrix}}_{\!\!\!\!\!\bold a_m}^{\!\!\!\!\!q_m}(x_m^{(R)},y_m^{(R)}),
{\begin{matrix} & \hskip-0.153in \triangleleft \\[-10.45pt]
\triangleleft & \hskip-0.153in\triangleright \\[-10.45pt]
& \hskip-0.153in\triangleleft \end{matrix}}_{\,\bold b_1}^{\,q'_1}(z_1^{(R)},w_1^{(R)}),
\dotsc,
{\begin{matrix} & \hskip-0.153in \triangleleft \\[-10.45pt]
\triangleleft & \hskip-0.153in\triangleright \\[-10.45pt]
& \hskip-0.153in\triangleleft \end{matrix}}_{\,\bold b_n}^{\,q'_n}(z_n^{(R)},w_n^{(R)})
\right)
}
-1
\nonumber
\\[10pt]
&
=
\left|\frac{\det\widetilde{M'_{\al,\be}}}{\det\widetilde{M'}}\right|-1.
\label{eec}
\end{align}

\parindent15pt

Since the points $(x_1,y_1),\dotsc,(z_n,w_n)$ are all distinct, the expression on the right hand side in Theorem \ref{tdb} is non-zero. This gives then the asymptotics of $\det \widetilde{M'}$ as $R\to\infty$. Since the only difference between $\det \widetilde{M'}_{\al,\be}$ and $\det \widetilde{M'}$ is that in the former the sequences $x_1^{(R)}$ and $y_1^{(R)}$ are replaced by $x_1^{(R)}+\al$ and $y_1^{(R)}+\be$ (which clearly still satisfy $(x_1^{(R)}+\al)/R\to x_1$, $(y_1^{(R)}+\be)/R\to y_1$), the same expression gives the asymptotics of $\det \widetilde{M'}_{\al,\be}$ as $R\to\infty$. It follows that for $R$ large enough we have
\begin{equation}
\left|\frac{\det\widetilde{M'}_{\al,\be}}{\det\widetilde{M'}}\right|-1
=
\frac{\det\widetilde{M'}_{\al,\be}}{\det\widetilde{M'}}-1
=
\frac{ \det\widetilde{M'}_{\al,\be} -\det\widetilde{M'}}{\det\widetilde{M'}}.
\label{eed}
\end{equation}  

Note that $\widetilde{M'}_{\al,\be}$ and $\widetilde{M'}$ are $2S\times 2S$ matrices with all but the first $2s_1$ rows the same. Applying Lemma \ref{tea} to them we get
\begin{equation}
\det\widetilde{M'}_{\al,\be} -\det\widetilde{M'}
=
\sum_{i=1}^{2s_1} \det M_i
\label{eee}
\end{equation}
where
\begin{equation}
M_i=
\left[\begin{matrix}
\text{Row 1 of $\widetilde{M'}$}\\
\cdot\\[-7pt]
\cdot\\[-7pt]
\cdot\\
\text{Row $i-1$ of $\widetilde{M'}$}\\[5pt]
\text{(Row $i$ of $\widetilde{M'}_{\al,\be}$)-(Row $i$ of $\widetilde{M'})$}\\[5pt]
\text{Row $i+1$ of $\widetilde{M'}_{\al,\be}$}\\
\cdot\\[-7pt]
\cdot\\[-7pt]
\cdot\\
\text{Row $2s_1$ of $\widetilde{M'}_{\al,\be}$}\\[5pt]
\text{Row $2s_1+1$ of $\widetilde{M'}$}\\
\cdot\\[-7pt]
\cdot\\[-7pt]
\cdot\\
\text{Row $2S$ of $\widetilde{M'}$}
\end{matrix}\right].
\label{eef}
\end{equation}
Because $3|\al-\be$, the main terms in the asymptotics of the entries in all rows except the $i$th row in $M_i$ are the same as the asymptotics of the corresponding entries of $\widetilde{M'}$ (this follows for the $P$-part from the $k=l=0$ specialization of equation \eqref{ege}, and for the $U$-part from the $k=0$ specialization of \eqref{egi}).
%

Furthermore, since we are assuming $3|\al-\be$, Proposition \ref{tgc} applies and gives that the asymptotics of each entry in row $i$ of $M_i$ is obtained by applying the operator $\al\frac{\partial}{\partial x_1}+\be\frac{\partial}{\partial y_1}$ to the main term in the corresponding entry of  $\widetilde{M'}$, and multiplying the result by $1/R$. By the well-known formula expressing the derivative of the determinant of an $n\times n$ matrix $A$ as the sum of the determinants of $n$ matrices, the $i$th of which is obtained from $A$ by replacing row $i$ by its derivative, we obtain
\begin{equation}
\sum_{i=1}^{2s_1} \det\lfloor\lfloor M_i\rfloor\rfloor=\frac{1}{R}\left(\al\frac{\partial}{\partial x_1}+\be\frac{\partial}{\partial y_1}\right)\det\lfloor\lfloor\widetilde{M'}\rfloor\rfloor.
\label{eeg}
\end{equation}
Then the fraction
%
\begin{equation*}
\frac{\sum_{i=1}^{2s_1} M_i}{\det\widetilde{M'}}
\end{equation*}
%
(which, by \eqref{eec}--\eqref{eee}, is equal to the left hand side of \eqref{eec}), satisfies
%
\begin{equation}
\frac{\sum_{i=1}^{2s_1} \det M_i}{\det\widetilde{M'}}
\sim
\frac{1}{R}
\frac
{\left(
\al\frac{\partial}{\partial x_1}+\be\frac{\partial}{\partial y_1}\right)\det\lfloor\lfloor\widetilde{M'}\rfloor\rfloor}
{\det\lfloor\lfloor\widetilde{M'}\rfloor\rfloor},\ \ \ R\to\infty.
\label{eegg}
\end{equation}
Using the explicit product formula for $\det\lfloor\lfloor\widetilde{M'}\rfloor\rfloor$ in Theorem \ref{tdb} one gets
\begin{align}
\frac
{\left(
\al\frac{\partial}{\partial x_1}+\be\frac{\partial}{\partial y_1}\right)\det\lfloor\lfloor\widetilde{M'}\rfloor\rfloor}
{\det\lfloor\lfloor\widetilde{M'}\rfloor\rfloor}
&=
\sum_{i=2}^m s_1s_i\frac{\al[2(x_1-x_i)+(y_1-y_i)]+\be[(x_1-x_i)+2(y_1-y_i)]}{(x_1-x_i)^2+(x_1-x_i)(y_1-y_i)+(y_1-y_i)^2}
\nonumber
\\[10pt]
&-\sum_{j=1}^n s_1t_j\frac{\al[2(x_1-z_j)+(y_1-w_j)]+\be[(x_1-z_j)+2(y_1-w_j)]}{(x_1-z_j)^2+(x_1-z_j)(y_1-w_j)+(y_1-w_j)^2}.
\end{align}  
The statement of Theorem \ref{tba} follows then using \eqref{eba} and the readily verified fact that 
if $(a,b)\neq(0,0)$ and $(x,y)$ are the coordinates of two vectors in our 60 degree system of coordinates, the orthogonal projection of $(x,y)$ on the direction of $(a,b)$ has length
\begin{equation*}
\hskip1.95in\proj_{(a,b)} (x,y) =  \frac{1}{2}\frac{a(2x+y)+b(x+2y)}{\sqrt{a^2+ab+b^2}}.\hskip1.95in\square 
\end{equation*}
%
%

\section{The general case}

To cover the general case, we need to prove the statement of Theorem \ref{tba} for $\al-\be=1\,(\mod 3)$ and for $\al-\be=-1\,(\mod 3)$. We present the details in the former case; the latter is similar. So throughout this section we assume that $\al-\be=1\,(\mod 3)$.

The general approach from Section 5 will also work in this case. However, since it was essential there that $\al-\be$ was a multiple of 3 (so that Proposition \ref{tgc} could be applied), we need to make some changes.

Since we cannot apply Proposition \ref{tgc} to the entries in row $i$ of the matrix $M_i$ resulting by applying Lemma \ref{tea} to the difference $\det {\widetilde M'}_{\al,\be}-\det {\widetilde M'}$, we need to go back to the original difference $\det M_{\al,\be}-\det M$ and apply first some convenient determinant-preserving row operations to the matrix $M_{\al,\be}$.

Namely, swap rows $R_{2i-1}$ and $R_{2i}$, and then replace $R_{2i}$ by $-R_{2i}-R_{2i-1}$, for $i=1,\dotsc,2s_1$ in $M_{\al,\be}$. These turn the $2\times2$ submatrix
\begin{equation}
\left[
\begin{matrix}
P\left({\!\!\!\!\!\!\!\!\!\!\!x_1^{(R)}-z_j^{(R)}+a_{1k}-b_{jl}+\al-1,\atop\ y_1^{(R)}-w_j^{(R)}+q_1a_{1k}-q'_jb_{jl}+\be-1\!}\right)
&
P\left({\!\!\!\!x_1^{(R)}-z_j^{(R)}+a_{1k}-b_{jl}+\al-2,\atop\ y_1^{(R)}-w_j^{(R)}+q_1a_{1k}-q'_jb_{jl}+\be\!}\right)
\\[20pt]
P\left({\!\!\!\!\!\!\!\!\!\!\!\!\!\!\!\!\!\!x_1^{(R)}-z_j^{(R)}+a_{1k}-b_{jl}+\al,\atop\ y_1^{(R)}-w_j^{(R)}+q_1a_{1k}-q'_jb_{jl}+\be-2\!}\right)
&
P\left({\!\!\!\!\!\!\!\!\!\!\!x_1^{(R)}-z_j^{(R)}+a_{1k}-b_{jl}+\al-1,\atop\ y_1^{(R)}-w_j^{(R)}+q_1a_{1k}-q'_jb_{jl}+\be-1\!}\right)
\end{matrix}
\right],
\label{efa}
\end{equation}
of $M_{\al,\be}$ (see \eqref{ecl}) into the $2\times2$ matrix with first column
\begin{equation}
\left[
\begin{matrix}
P\left({\!\!\!\!\!\!\!\!\!\!\!x_1^{(R)}-z_j^{(R)}+a_{1k}-b_{jl}+(\al+1)-1,\atop\ y_1^{(R)}-w_j^{(R)}+q_1a_{1k}-q'_jb_{jl}+(\be-1)-1\!}\right)
\\[20pt]
-P\left({\!\!\!\!\!\!\!\!\!\!\!\!\!\!\!x_1^{(R)}-z_j^{(R)}+a_{1k}-b_{jl}+\al-1,\atop\ y_1^{(R)}-w_j^{(R)}+q_1a_{1k}-q'_jb_{jl}+\be-1\!}\right)
-P\left({\!\!\!\!\!\!\!\!\!\!\!\!\!\!\!x_1^{(R)}-z_j^{(R)}+a_{1k}-b_{jl}+(\al+1)-1,\atop\ y_1^{(R)}-w_j^{(R)}+q_1a_{1k}-q'_jb_{jl}+(\be-1)-1\!}\right)
\end{matrix}
\right]
\label{efb}
\end{equation}
and second column
\begin{equation}
\left[
\begin{matrix}
P\left({\!\!\!\!\!\!\!\!\!\!\!\!x_1^{(R)}-z_j^{(R)}+a_{1k}-b_{jl}+\al-1,\atop\ y_1^{(R)}-w_j^{(R)}+q_1a_{1k}-q'_jb_{jl}+\be-1\!}\right)
\\[20pt]
-P\left({\!\!\!\!\!\!\!\!\!\!\!x_1^{(R)}-z_j^{(R)}+a_{1k}-b_{jl}+(\al-1)-1,\atop\ y_1^{(R)}-w_j^{(R)}+q_1a_{1k}-q'_jb_{jl}+(\be+1)-1\!}\right)
-P\left({\!\!\!\!\!\!\!\!\!\!\!x_1^{(R)}-z_j^{(R)}+a_{1k}-b_{jl}+\al-1,\atop\ y_1^{(R)}-w_j^{(R)}+q_1a_{1k}-q'_jb_{jl}+\be-1\!}\right)
\end{matrix}
\right].
\label{efc}
\end{equation}
However, the coupling function $P$ satisfies $P(x,y)+P(x-1,y)+P(x,y-1)=0$ for all $(x,y)\neq(0,0)$ (see \cite{KeLocal}). Therefore, the $2\times2$ block obtained from \eqref{efa} by the above row operations is 
\begin{equation}
\left[
\begin{matrix}
P\left({\!\!\!\!\!\!\!\!\!\!\!x_1^{(R)}-z_j^{(R)}+a_{1k}-b_{jl}+(\al+1)-1,\atop\ y_1^{(R)}-w_j^{(R)}+q_1a_{1k}-q'_jb_{jl}+(\be-1)-1\!}\right)
&
P\left({\!\!\!\!\!x_1^{(R)}-z_j^{(R)}+a_{1k}-b_{jl}+(\al+1)-2,\atop\ y_1^{(R)}-w_j^{(R)}+q_1a_{1k}-q'_jb_{jl}+(\be-1)\!}\right)
\\[20pt]
P\left({\!\!\!\!\!\!\!\!\!\!\!\!\!\!\!\!\!\!\!\!\!\!\!\!\!\!\!\!\!\!\!\!\!x_1^{(R)}-z_j^{(R)}+a_{1k}-b_{jl}+\al,\atop\ y_1^{(R)}-w_j^{(R)}+q_1a_{1k}-q'_jb_{jl}+(\be+1)-2\!}\right)
&
P\left({\!\!\!\!\!\!\!\!\!\!\!\!\!\!\!\!\!\!\!\!\!\!\!x_1^{(R)}-z_j^{(R)}+a_{1k}-b_{jl}+\al-1,\atop\ y_1^{(R)}-w_j^{(R)}+q_1a_{1k}-q'_jb_{jl}+(\be+1)-1\!}\right)
\end{matrix}
\right].
\label{efd}
\end{equation}
Similarly, using the identity of $U_s(a,b)+U_s(a-1,b)+U_s(a,b-1)=0$ (which follows from \eqref{ecb} and $P(x,y)+P(x-1,y)+P(x,y-1)=0$), the $2\times2$ submatrix
\begin{equation}
\left[
\begin{matrix}
U_{l-1}\left({\!\!\!\!\!x_1^{(R)}+a_{1k}+\al,\atop\ y_1^{(R)}+q_1a_{1k}+\be\!}\right)
&
U_{l-1}\left({\!\!\!\!\!x_1^{(R)}+a_{1k}+\al-1,\atop\ y_1^{(R)}+q_1a_{1k}+\be+1\!}\right)
\\[10pt]
U_{l-1}\left({\!\!\!\!\!x_1^{(R)}+a_{1k}+\al+1,\atop\ y_1^{(R)}+q_1a_{1k}+\be-1\!}\right)
&
U_{l-1}\left({\!\!\!\!\!x_1^{(R)}+a_{1k}+\al,\atop\ y_1^{(R)}+q_1a_{1k}+\be\!}\right)
\end{matrix}\right]
\label{efe}
\end{equation}
of $M_{\al,\be}$ (see \eqref{ecp}) is transformed by the above row operations into
\begin{equation}
\left[
\begin{matrix}
U_{l-1}\left({\!\!\!\!\!x_1^{(R)}+a_{1k}+(\al+1),\atop\ y_1^{(R)}+q_1a_{1k}+(\be-1)\!}\right)
&
U_{l-1}\left({\!\!\!\!\!x_1^{(R)}+a_{1k}+(\al+1)-1,\atop\ y_1^{(R)}+q_1a_{1k}+(\be-1)+1\!}\right)
\\[10pt]
U_{l-1}\left({\!\!\!\!\!\!\!\!\!\!\!\!\!\!\!\!\!\!\!x_1^{(R)}+a_{1k}+\al+1,\atop\ y_1^{(R)}+q_1a_{1k}+(\be+1)-1\!}\right)
&
U_{l-1}\left({\!\!\!\!\!\!\!\!\!\!\!\!\!\!\!\!\!x_1^{(R)}+a_{1k}+\al,\atop\ y_1^{(R)}+q_1a_{1k}+(\be+1)\!}\right)
\end{matrix}\right].
\label{eff}
\end{equation}
The great advantage of the transformed matrix --- let us denote it by $N_{\al,\be}$ --- is that its $(2i-1)$st row is the $(2i-1)$st row of $M_{\al+1,\be-1}$, and its $2i$th row is the $2i$th row of $M_{\al,\be+1}$, for $i=1,\dotsc,s_1$
--- and both $(\al+1)-(\be-1)$ and $\al-(\be+1)$ are multiples of 3 (as we are assuming $\al-\be=1\,(\mod 3)$), so we can now apply Proposition \ref{tgc}. 

Since the entries of $N_{\al,\be}$ have a uniform definition along odd-index rows, and also along even-index rows, the row and column operations described in Section 3 can be applied to it. Let  ${\widetilde N'}_{\al,\be}$ be the matrix obtained from $N_{\al,\be}$ by the procedure that transformed $M_{\al,\be}$ into ${\widetilde M'}_{\al,\be}$.

Following the arguments that led to \eqref{eee} (and using also that the row operations that turned $M_{\al,\be}$ into $N_{\al,\be}$ preserve the determinant), we obtain
\begin{align}
\det M_{\al,\be} -\det M
&=
\det N_{\al,\be} -\det M
\nonumber
\\[10pt]
&=
\det\widetilde{N'}_{\al,\be} -\det\widetilde{M'}
\nonumber
\\[10pt]
&=
\sum_{i=1}^{s_1} \left(\det M_{2i-1}^{\al+1,\be-1}+\det M_{2i}^{\al,\be+1}\right),
\label{efg}
\end{align}
where $M_i^{\al,\be}$ is the matrix given by \eqref{eef} (denoted in Section 5 simply by $M_i$).

The arguments that proved \eqref{eegg} give then
\begin{align}
&
\frac{\sum_{i=1}^{s_1} \det M_{2i-1}^{\al+1,\be-1}+\det M_{2i}^{\al,\be+1}}{\det\widetilde{M'}}
\nonumber
\\[10pt]
&\ 
\sim
\frac{1}{R}
\frac
{\sum_{i=1}^{s_1} \left\{\det\left((\al+1)\frac{\partial}{\partial x_1}+(\be-1)\frac{\partial}{\partial y_1}\right)^{(2i-1)}
\lfloor\lfloor\widetilde{M'}\rfloor\rfloor
+\det\left(\al\frac{\partial}{\partial x_1}+(\be+1)\frac{\partial}{\partial y_1}\right)^{(2i)}
\lfloor\lfloor\widetilde{M'}\rfloor\rfloor\right\}}
{\det\lfloor\lfloor\widetilde{M'}\rfloor\rfloor},
\nonumber
\\[10pt]
&\ \ \ \ \ \ \ \ \ \ \ \ \ \ \ \ \ \ \ \ \ \ \ \ \ \ \ \ \ \ \ \ \ \ \ \ \ \ \ \ \ \ \ \ \ \ \ \ \ \ \ \ \ \ \ \ \ \ \ \ \ \ \ \ \ \ \ \ \ \ \ \ \ \ \ \ \ \ \ \ \ \ \ \ \ \ \ \ \ \ \ \ \ \ \ \ \ \ \ \ \ \ \ \ \ \ \ \ \ \ \ \ 
R\to\infty,
\label{efh}
\end{align}
where for an operator $F$ that acts on the entries of a matrix $M$, $F^{(i)} M$ stands for the matrix obtained from $M$ by applying $F$ to all the elements in row $i$.

Using the linearity of the determinant in the rows, we have
\begin{align}
&
\sum_{i=1}^{s_1} \left\{\det\left((\al+1)\frac{\partial}{\partial x_1}+(\be-1)\frac{\partial}{\partial y_1}\right)^{(2i-1)}
\lfloor\lfloor\widetilde{M'}\rfloor\rfloor
+\det\left(\al\frac{\partial}{\partial x_1}+(\be+1)\frac{\partial}{\partial y_1}\right)^{(2i)}
\lfloor\lfloor\widetilde{M'}\rfloor\rfloor\right\}
\nonumber
\\[10pt]
&
=
\sum_{i=1}^{s_1} \left\{\det\left(\al\frac{\partial}{\partial x_1}+\be\frac{\partial}{\partial y_1}\right)^{(2i-1)}
\lfloor\lfloor\widetilde{M'}\rfloor\rfloor
+\det\left(\frac{\partial}{\partial x_1}-\frac{\partial}{\partial y_1}\right)^{(2i-1)}
\lfloor\lfloor\widetilde{M'}\rfloor\rfloor\right.
\nonumber
\\[10pt]
&\ \ \ \ \ \ \ \ \ \ \ \ \ \ \ \ \ \ \ \ \ \ \ \ \ \ \ \ \ \ \ \ \ \ \ \ \ \ \ \ \ \ \,
+
\left. \det\left(\al\frac{\partial}{\partial x_1}+\be\frac{\partial}{\partial y_1}\right)^{(2i)}
\lfloor\lfloor\widetilde{M'}\rfloor\rfloor
+\det\left(\frac{\partial}{\partial y_1}\right)^{(2i)}
\lfloor\lfloor\widetilde{M'}\rfloor\rfloor\right\}
\nonumber
\\[10pt]
&
=
\left(\al\frac{\partial}{\partial x_1}+\be\frac{\partial}{\partial y_1}\right)
\det\lfloor\lfloor\widetilde{M'}\rfloor\rfloor
+
\sum_{i=1}^{s_1} \left\{\det\left(\frac{\partial}{\partial x_1}-\frac{\partial}{\partial y_1}\right)^{(2i-1)}
\lfloor\lfloor\widetilde{M'}\rfloor\rfloor
+\det\left(\frac{\partial}{\partial y_1}\right)^{(2i)}
\lfloor\lfloor\widetilde{M'}\rfloor\rfloor\right\}.
\label{efi}
\end{align}
We claim that
\begin{equation}
\det\left(\frac{\partial}{\partial x_1}-\frac{\partial}{\partial y_1}\right)^{(2i-1)}
\lfloor\lfloor\widetilde{M'}\rfloor\rfloor
+\det\left(\frac{\partial}{\partial y_1}\right)^{(2i)}
\lfloor\lfloor\widetilde{M'}\rfloor\rfloor
=0,
\label{efj}
\end{equation}
for $i=1,\dotsc,s_1$. Once we prove this, using \eqref{efg}--\eqref{efi}, the statement of Theorem \ref{tba} follows by the arguments of the case $3|\al-\be$.

For ease of reference, let
\begin{align}
C_1&=\left(\frac{\partial}{\partial x_1}-\frac{\partial}{\partial y_1}\right)^{(2i-1)}\lfloor\lfloor\widetilde{M'}\rfloor\rfloor
\label{efk}
\\[10pt]
C_2&=\left(\frac{\partial}{\partial y_1}\right)^{(2i)}\lfloor\lfloor\widetilde{M'}\rfloor\rfloor.
\label{efl}
\end{align}
To complete the proof, we need to show that $\det C_1 +\det C_2 =0$.

Recall the structure of the matrix $\lfloor\lfloor\widetilde{M'}\rfloor\rfloor$, whose entries are given by equations \eqref{edd}--\eqref{edf}: It's ``$A$''-part is an $m\times n$ block matrix, with blocks given by \eqref{ede}, and its ``$B$''-part is an $m\times1$ block matrix with blocks given by \eqref{edf}.

By definition, along all rows except row $2i-1$, $C_1$ is the same as $\lfloor\lfloor\widetilde{M'}\rfloor\rfloor$, whose entries are given by equations \eqref{edd}--\eqref{edf}; its $(2i-1)$st row is obtained by applying $\frac{\partial}{\partial x_1}-\frac{\partial}{\partial y_1}$ to row $2i-1$ of $\lfloor\lfloor\widetilde{M'}\rfloor\rfloor$. Similarly, along all rows except row $2i$, $C_2$ is the same as $\lfloor\lfloor\widetilde{M'}\rfloor\rfloor$, while its $2i$th row is obtained by applying $\frac{\partial}{\partial y_1}$ to row $2i$ of $\lfloor\lfloor\widetilde{M'}\rfloor\rfloor$.

The expressions for the entries in the exceptional rows of $C_1$ and $C_2$ are readily obtained using formulas \eqref{egee} (with $u=x_1-z_j$, $v=y_1-w_j$ for the entries in the positions occupied by block matrix $(1,j)$ in the $A$-part) and \eqref{egii}.
They imply that corresponding $2\times2$ blocks of $C_1$ and $C_2$ along rows $2i-1$ and $2i$ are of the form
\begin{equation}
\left[\begin{matrix}  
\left\langle(1+\z)\z^{e}g(\z)\right\rangle & \left\langle(1+\z)\z^{e-2}g(\z)\right\rangle\\[15pt]
\left\langle\z^{e+2}f(\z)\right\rangle & \left\langle\z^{e}f(\z)\right\rangle
\end{matrix}\right]    
\ \ \ \text{\rm and}\ \ \ 
\left[\begin{matrix}  
\left\langle\z^{e}f(\z)\right\rangle & \left\langle\z^{e-2}f(\z)\right\rangle\\[15pt]
\left\langle(-\z)\z^{e+2}g(\z)\right\rangle & \left\langle(-\z)\z^{e}g(\z)\right\rangle
\end{matrix}\right],
\label{efm}
\end{equation}
respectively, with the same exponents $e$ and the same functions $f$ and $g$ (because the only reflection of $\al$ and $\be$ on the expressions on the right hand sides of equations \eqref{egee} and \eqref{egii} is in the multiplicative factors $\al-\z\be$, shown in parentheses in the matrix entries above).

Since $1+\z=-\z^{-1}$ and $\z^3=1$, the blocks \eqref{efm} can be rewritten as
\begin{equation}
\left[\begin{matrix}  
-\left\langle\z^{e-1}g(\z)\right\rangle & -\left\langle\z^{e}g(\z)\right\rangle\\[15pt]
\left\langle\z^{e-1}f(\z)\right\rangle & \left\langle\z^{e}f(\z)\right\rangle
\end{matrix}\right]    
\ \ \ \text{\rm and}\ \ \ 
\left[\begin{matrix}  
\left\langle\z^{e}f(\z)\right\rangle & \left\langle\z^{e+1}f(\z)\right\rangle\\[15pt]
-\left\langle\z^{e}g(\z)\right\rangle & -\left\langle\z^{e+1}g(\z)\right\rangle
\end{matrix}\right].
\label{efn}
\end{equation}
Let $C'_1$ be the matrix obtained from $C_1$ by multiplying row $2i-1$ by $-1$, and then swapping rows $2i-1$ and $2i$. Then $\det C_1=\det C'_1$, and corresponding blocks of $C'_1$ and $C_2$ along rows $2i-1$ and $2i$ are of the form
\begin{equation}
\left[\begin{matrix}  
\left\langle\z^{e-1}f(\z)\right\rangle & \left\langle\z^{e}f(\z)\right\rangle\\[15pt]
\left\langle\z^{e-1}g(\z)\right\rangle & \left\langle\z^{e}g(\z)\right\rangle
\end{matrix}\right]    
\ \ \ \text{\rm and}\ \ \ 
\left[\begin{matrix}  
\left\langle\z^{e}f(\z)\right\rangle & \left\langle\z^{e+1}f(\z)\right\rangle\\[15pt]
-\left\langle\z^{e}g(\z)\right\rangle & -\left\langle\z^{e+1}g(\z)\right\rangle
\end{matrix}\right].
\label{efo}
\end{equation}
Now regard the set of $2S$ columns of $C'_1$ as consisting of $S$ strips of pairs of consecutive columns, and in each such strip perform two column operations on $C'_1$: (1) swap the two columns in the strip, and (2) replace the second column in the strip by its negative minus the first column. Using that $-\z^e-\z^{e-1}=\z^{e+1}$, one readily sees that the effect on the $2\times2$ block of $C'_1$ on the left in \eqref{efo} is to turn it into 
\begin{equation}
\left[\begin{matrix}  
\left\langle\z^{e}f(\z)\right\rangle & \left\langle\z^{e+1}f(\z)\right\rangle\\[15pt]
\left\langle\z^{e}g(\z)\right\rangle & \left\langle\z^{e+1}g(\z)\right\rangle
\end{matrix}\right],  
\label{efp}
\end{equation}
which is precisely the same as the $2\times2$ block of $C_2$ on the right in \eqref{efo}, with the second row multiplied by $-1$.

However, the resulting matrix $C''_1$ --- which has $\det C''_1=\det C'_1$, as each of the column operations (1) and (2) has the effect of multiplying the determinant by $(-1)$ --- no longer agrees with $C_1$ outside rows $2i-1$ and $2i$ (as $C_1$ and $C'_1$ do), because the column operations we applied to arrive at it changed the entries there. 

This can nevertheless be easily remedied, using row operations. Indeed, each $2\times2$ block outside rows $2i-1$ and $2i$ in $C'_1$ has form
\begin{equation}
\left[\begin{matrix}  
\left\langle\z^{e}f(\z)\right\rangle & \left\langle\z^{e-2}f(\z)\right\rangle\\[15pt]
\left\langle\z^{e+2}f(\z)\right\rangle & \left\langle\z^{e}f(\z)\right\rangle
\end{matrix}\right] 
\label{efq}
\end{equation}
for some integer $e$ and function $f$. After applying column operations (1) and (2) above it becomes
\begin{equation}
\left[\begin{matrix}  
\left\langle\z^{e-2}f(\z)\right\rangle & \left\langle\z^{e-1}f(\z)\right\rangle\\[15pt]
\left\langle\z^{e}f(\z)\right\rangle & \left\langle\z^{e+1}f(\z)\right\rangle
\end{matrix}\right].
\label{efr}
\end{equation}
In order to restore it to form \eqref{efq}, swap its rows and then replace the second row by its negative minus the first row. Since the column operations (1) and (2) were applied to all $S$ pairs of consecutive columns, these row operations restore all the entries of $C''_1$ outside rows $2i-1$ and $2i$ to their original form.

Therefore, if we denote by $C'''_1$ the resulting matrix, we obtain that $C'''_1$ agrees with $C_2$ along all rows except row $2i$, where its entries are the negatives of the corresponding entries of $C_2$. Therefore $\det C'''_1=-\det C_2$, and since $\det C'''_1=\det C''_2$, we obtain $\det C_1+\det C_2=0$. This completes the proof of Theorem \ref{tba}.
%
%
%

\section{The asymptotics of the entries of the $i$th row of $M_i$ when $3|\al-\be$}

In this section we prove that when $\al-\be$ is a multiple of 3, the $R\to\infty$ asymptotics of each entry in the $i$th row of the matrix $M_i$ defined by \eqref{eef} is obtained by applying the operator $\al\frac{\partial}{\partial x_1}+\be\frac{\partial}{\partial y_1}$ to the corresponding entry of the $i$th row of $\widetilde{M'}$, and multiplying the result by $1/R$.

Since the entries of $\widetilde{M'}$ are defined differently in columns $1,\dotsc,2T$ compared to columns $2T+1,\dotsc,2S$ (see \eqref{eda}), and since the same clearly holds for $\widetilde{M'}_{\al,\be}$ (as $\widetilde{M'}_{\al,\be}$ is obtained from the matrix in \eqref{eda} by simply replacing $x_1^{(R)}$ by $x_1^{(R)}+\al$ and $y_1^{(R)}$ by $y_1^{(R)}+\be)$, the assertion of the previous paragraph needs to be checked separately for the first $2T$ entries in row $i$ of $M_i$ and for its last $2S-2T$ entries. We do this in the next two lemmas.

Recall that ${\mathcal D}$ denotes Newton's divided difference operator, whose powers are defined inductively by 
${\mathcal D}^0 f=f$ and ${\mathcal D}^r f(c_j)=({\mathcal D}^{r-1} f(c_{j+1})-{\mathcal D}^{r-1} f(c_j))/(c_{j+r}-c_j)$, $r\geq1$.
We will need the following result on the asymptotics of the coupling function $P$ when acted on in the
indicated way by powers of ${\mathcal D}$. 

Given a two-variable function $G$, denote by $\Delta_{\al,\be}$ the difference operator
\begin{equation*}
\Delta_{\al,\be}\,G(x,y)=G(x+\al,y+\be)-G(x,y).
\end{equation*}  
\begin{lem}        
\label{tga}
Let $r_n$ and $s_n$ be integers so that $\lim_{n\to\infty}r_n/n=u$, $\lim_{n\to\infty}s_n/n=v$, and $(u,v)\neq(0,0)$. Then for any integers $k,l\geq0$, any rational numbers $q,q'$ with $3|1-q,1-q'$ and any integers $c,d,\al,\be$ with $3|\al-\be$ we have as $n\to\infty$ that\footnote{ For a function $f(n)$, $\lfloor\lfloor f(n)\rfloor\rfloor$ stands for the dominant part of $f(n)$ as $\to\infty$.}
\begin{align}
&
\lfloor\lfloor\Delta_{\al,\be}\,\left.{\mathcal D}_y^l\left\{{\mathcal D}_x^k\, P(r_n+x+y+c,s_n+qx+q'y+d))|_{x=a_1}\right\}\right|_{y=b_1}\rfloor\rfloor=
\nonumber
\\[10pt]
&\ \ \ \ \ \ \ \ \ \ \ \ \ \ \ \ \  
\frac{1}{n}\left(\al\frac{\partial}{\partial u}+\be\frac{\partial}{\partial v}\right)
\lfloor\lfloor\left.{\mathcal D}_y^l\left\{{\mathcal D}_x^k\, P(r_n+x+y+c,s_n+qx+q'y+d))|_{x=a_1}\right\}\right|_{y=b_1}\rfloor\rfloor,
\label{ega}
\end{align}
where ${\mathcal D}^k_x$ acts with respect to a fixed integer sequence $a_1,a_2,\dotsc$, and
${\mathcal D}^l_y$ acts with respect to an integer sequence $b_1,b_2,\dotsc$ satisfying $qb_j\in\Z$ for all
$j\geq1$. 

\end{lem}

\begin{proof}
Suppose first that $u<0$. Then by the arguments in the proof of Theorem 4.1 in \cite{ef} we obtain
\begin{align}
&
\Delta_{\al,\be}\,\left.{\mathcal D}_y^l\left\{{\mathcal D}_x^k\, P(r_n+x+y+c,s_n+qx+q'y+d))|_{x=a_1}\right\}\right|_{y=b_1}
\nonumber
\\
&
=
\frac{1}{2\pi i}
\left\langle
  \int_\z^{-1}(-1-t)^{-r_n-c}t^{-s_n-d}[(-1-t)^{-\al}t^{-\be}-1](t-\z)^{k+l}
\right.
\nonumber
\\
&\ \ \ \ \ \ \ \ \ \ \ \ \ \ \
\left.\left.
\times
\left\{
    \frac{1}{k!\,l!}(\z-q\z^{-1})^{k}(\z-q'\z^{-1})^{l}(t-\z)^{k+l}
+c_{k+l+1}(t-\z)^{k+l+1}+\cdots
  \right\}dt
\right\rangle.\right.
\label{egb}
\end{align}
This integral differs from the one in the proof of Theorem 4.1 in \cite{ef} only by the presence of the extra factor $(-1-t)^{-c}t^{-d}[(-1-t)^{-\al}t^{-\be}-1]$ in the integrand\footnote{ The effect of having the factor $(\z-q\z^{-1})^{k}(\z-q'\z^{-1})^{l}$ in \eqref{egb} instead of the corresponding $(\z-q\z^{-1})^{k+l}$ in \cite{ef} is automatically taken into account by the form of the expansion in the curly braces in \eqref{egb}.}. The way this changes the asymptotics is determined by the form of the series expansion of this extra factor around $t=\z$ (see \cite[\S 6.4]{Olver}). Since $3|\al-\be$, this series expansion is
\begin{align}
  (-1-t)^{-c}t^{-d}[(-1-t)^{-\al}t^{-\be}-1]&=(\zeta^{c-d}+c_1(t-\zeta)+\cdots)
\nonumber
  \\[10pt]
&\ \ \ \ \ \times
  (\zeta^{\al-\be}-1+(\al\zeta-\be/\z)(t-\z)+c'_2(t-\zeta)^2+\cdots)
\nonumber
\\[10pt]
&=\zeta^{c-d}(\al\zeta-\be/\z)(t-\z)+\cdots.
\label{egc}
\end{align}
By \cite[Proposition 4.4]{ef}, it follows that the integral in \eqref{egb} is 
\begin{equation}
\frac{(k+l+1)!}{k!\,l!}\frac{\z^{r_n-s_n+c-d}(\al\z-\be\z^{-1})(\z-q\z^{-1})^k(\z-q'\z^{-1})^l}{(-r_n\z+s_n\z^{-1})^{k+l+2}}
+O\left(\frac{1}{n^{k+l+3}}\right).
\label{egd}
\end{equation}
Therefore, by \eqref{egb} we get
\begin{align}
&
\Delta_{\al,\be}\,\left.{\mathcal D}_y^l\left\{{\mathcal D}_x^k\, P(r_n+x+y+c,s_n+qx+q'y+d))|_{x=a_1}\right\}\right|_{y=b_1}
\nonumber
\\[10pt]
&
=
\frac{1}{2\pi i}\frac{(k+l+1)!}{k!\,l!}
\left\langle
\frac{\z^{r_n-s_n+c-d}(\al\z-\be\z^{-1})(\z-q\z^{-1})^k(\z-q'\z^{-1})^l}{(-r_n\z+s_n\z^{-1})^{k+l+2}}
\right\rangle+O\left(\frac{1}{n^{k+l+2}}\right)
\nonumber
\\[10pt]
&
\sim
\frac{1}{2\pi i}\frac{(k+l+1)!}{k!\,l!}
\left\langle
\frac{\z^{r_n-s_n+c-d-1}(\al-\be\z)(1-q\z)^k(1-q')^l}{(-u+v\z)^{k+l+2}}
\right\rangle
\frac{1}{n^{k+l+2}}.
\label{egdd}
\end{align}
On the other hand, by a straightforward extension of Theorem 4.1 of \cite{ef} (see \cite[\S14, Proposition 7.1']{ec}) we have
\begin{align}
&
\left.{\mathcal D}^l_y\left\{{\mathcal D}^k_x \,P(r_n+x+y+c,s_n+qx+q'y+d)|_{x=a_1}\right\}\right|_{y=b_1}=
\nonumber
\\[10pt]
&\ \ \ \ \ \ \ \ \ \ \ \ \ \ \ \ \ \ \ \ 
\frac{1}{2\pi i}{k+l\choose k}
\left\langle\frac{\zeta^{r_n-s_n+c-d-1}(1-q\z)^{k}(1-q'\z)^{l}}{(-r_n+s_n\z)^{k+l+1}}\right\rangle
+O\left(\frac{1}{n^{k+l+2}}\right).
\nonumber
\\[10pt]
&\ \ \ \ \ \ \ \ \ \ \ \ \ \ \ \ \ \ \ \ 
\sim
\frac{1}{2\pi i}{k+l\choose k}
\left\langle\frac{\zeta^{r_n-s_n+c-d-1}(1-q\z)^{k}(1-q'\z)^{l}}{(-u+v\z)^{k+l+1}}\right\rangle
\frac{1}{n^{k+l+1}}.
\label{ege}
\end{align}

We need to check that
the main term in \eqref{egdd} 
is equal to $1/n$ times the result of applying the operator $\al\frac{\partial}{\partial u}+\be\frac{\partial}{\partial v}$ to the main term in \eqref{ege}; i.e., that
\begin{align}
&
\frac{1}{n}
\left(\al\frac{\partial}{\partial u}+\be\frac{\partial}{\partial v}\right)
\frac{1}{2\pi i}{k+l\choose k}
\left\langle\frac{\zeta^{r_n-s_n+c-d-1}(1-q\z)^{k}(1-q'\z)^{l}}{(-u+v\z)^{k+l+1}}\right\rangle
\frac{1}{n^{k+l+1}}
\nonumber
\\[10pt]
&\ \ \ \ \ 
=
\frac{1}{2\pi i}\frac{(k+l+1)!}{k!\,l!}
\left\langle
\frac{\z^{r_n-s_n+c-d-1}(\al-\be\z)(1-q\z)^k(1-q'\z)^l}{(-u+v\z)^{k+l+2}}
\right\rangle
\frac{1}{n^{k+l+2}}.
\label{egee}
\end{align}

This is readily checked.

Since $(u,v)\neq(0,0)$, at least one of $u<0$, $v<0$, and $-u-v<0$ is true. The 
symmetries $P(u,v)=P(-u-v-1,u)$ and $P(u,v)=P(v,u)$
of the coupling function allow one to use the same arguments that proved the case $u<0$ to deduce the 
other two cases (see the proof of Proposition 7.1 in \cite{ec} for details).
\end{proof}

\begin{lem}        
\label{tgb}
Let $x_1^{(n)},y_1^{(n)}\in\Z$ so that $\lim_{n\to\infty}x_1^{(n)}/n=x_1$ and $\lim_{n\to\infty}y_1^{(n)}/n=y_1$. Then for any integer $k\geq0$, any rational number $q$ with $3|1-q$ and any integers $c,d,\al,\be$ with $3|\al-\be$ we have as $n\to\infty$ that
\begin{align}
&
\lfloor\lfloor\Delta_{\al.\be}\,{\mathcal D}_x^k\, U_l(x_1^{(n)}+x+c,y_1^{(n)}+qx+d))|_{x=a_1}\rfloor\rfloor=
\nonumber
\\[10pt]
&\ \ \ \ \ \ \ \ \ \ \ \ \ \ \ \ \  
\frac{1}{n}\left(\al\frac{\partial}{\partial x_1}+\be\frac{\partial}{\partial y_1}\right)
\lfloor\lfloor{\mathcal D}_x^k\, U_l(x_1^{(n)}+x+c,y_1^{(n)}+qx+d))|_{x=a_1}\rfloor\rfloor,
\label{egf}
\end{align}
where ${\mathcal D}^k_x$ acts with respect to a fixed integer sequence $a_1,a_2,\dotsc$.


\end{lem}

\begin{proof}
\ \ \
By \cite[(6.8)]{ec} one has
\begin{align*}
U_l(a,b)=&\frac{1}{2\pi i}
\left\langle\zeta^{a-b-1}(a-b\zeta)^l\right\rangle
\nonumber
\\[10pt]
&+\ {\text{\rm monomials in $a$ and $b$ of joint degree $<l$}}.
\end{align*}
This implies, since $3|1-q$, that
\begin{align}
&
U_l(x_1^{(n)}+x+c,y_1^{(n)}+qx+d)=\frac{1}{2\pi i}
\left\langle
  \z^{x_1^{(n)}-y_1^{(n)}+c-d-1}[(1-q\z)x+x_1^{(n)}-y_1^{(n)}\z+c-d\z)]^l
\right\rangle
\nonumber
\\[10pt]
&
+\sum_{\mu,\nu\geq0 \atop \mu+\nu<l}c_{\mu,\nu}(x_1^{(n)}+x+c)^{\mu}(y_1^{(n)}+qx+d)^{\nu},
\label{egg}
\end{align}
where $c_{\mu,\nu}$ is independent of $x_1^{(n)}$ and $y_1^{(n)}$ for all $\mu$ and $\nu$.

Moreover, the argument that proved \cite[Lemma 6.4]{ec} implies that for any constants
$A,B,C,D\in\C$
\begin{align}
{\mathcal D}^k_x(Ax+Bx_1^{(n)}+Cy_1^{(n)}+D)^l|_{x=a_1}&={l\choose k}A^k\left(Bx_1^{(n)}+Cy_1^{(n)}\right)^{l-k}+O\left(n^{l-k-1}\right)
\nonumber
\\[10pt]
&={l\choose k}A^k\left(Bx_1+Cy_1\right)^{l-k}n^{l-k}+O\left(n^{l-k-1}\right).
\label{egh}
\end{align}
The above two equations imply\footnote{ We write $(x_1-y_1\z)^{l-k}$ as $\frac{1}{(x_1-y_1\z)^{k-l}}$ to economize horizontal space.}
\begin{align}
{\mathcal D}^k_xU_l(x_1^{(n)}+x+c,y_1^{(n)}+qx+d)|_{x=a_1}
=
\frac{1}{2\pi i}{l\choose k}
\left\langle \frac{\z^{x_1^{(n)}-y_1^{(n)}+c-d-1}(1-q\z)^k}{(x_1-y_1\z)^{k-l}}\right\rangle n^{l-k}+O\left(n^{l-k-1}\right).
\label{egi}
\end{align}
Since $3|\al-\be$, \eqref{egg} implies
\begin{align}
&
U_l(x_1^{(n)}+x+c+\al,y_1^{(n)}+qx+d+\be)=
\nonumber
\\[10pt]
&\ \ \ \ \ \ \ \ \ \ \ \ \ \ \ \ \ \ \ \ \ \ \ \ \ \ \ \ \ \ 
\frac{1}{2\pi i}
\left\langle
  \z^{x_1^{(n)}-y_1^{(n)}+c-d-1}[(1-q\z)x+x_1^{(n)}-y_1^{(n)}\z+c-d\z+\al-\be\z)]^l
\right\rangle
\nonumber
\\[10pt]
&\ \ \ \ \ \ \ \ \ \ \ \ \ \ \ \ \ \ \ \ \ \ \ \ \ \ \ \ \ \ 
+\sum_{\mu,\nu\geq0 \atop \mu+\nu<l}c_{\mu,\nu}(x_1^{(n)}+x+c+\al)^{\mu}(y_1^{(n)}+qx+d+\be)^{\nu}.
\label{egj}
\end{align}
We claim that when we apply ${\mathcal D}^k$ to the sum in \eqref{egg} we get
\begin{equation}
{\mathcal D}^k\left.\left\{\sum_{\mu,\nu\geq0,\, \mu+\nu<l}c_{\mu,\nu}(x_1^{(n)}+x+c)^{\mu}(y_1^{(n)}+qx+d)^{\nu}\right\}\right|_{x=a_1}
=\lambda n^{l-k-1} + O(n^{l-k-2}),
\label{egk}
\end{equation}
with $\lambda$ independent of $c$ and $d$. Indeed, this follows by expanding the summand in powers of $x$ and applying Lemma 6.3 of \cite{ec}.

Apply ${\mathcal D}^k$ to equations \eqref{egg} and \eqref{egj} and take the difference, so as to obtain
\begin{equation}
\Delta_{\al.\be}\,{\mathcal D}_x^k\, U_l(x_1^{(n)}+x+c,y_1^{(n)}+qx+d))|_{x=a_1}
\label{ebl}
\end{equation}
on the left hand side. Because of \eqref{egk}, ${\mathcal D}^k$ applied to the sums on the right hand sides of \eqref{egg} and \eqref{egj} have the same dominant term, and it is of order $n^{l-k-1}$. Therefore, when taking the difference they cancel. It follows that the dominant term in \eqref{ebl} is the same as the dominant term in
\begin{align}
&
  {\mathcal D}^k\left.\left\{
\frac{1}{2\pi i}
\left\langle
  \z^{x_1^{(n)}-y_1^{(n)}+c-d-1}[(1-q\z)x+x_1^{(n)}-y_1^{(n)}\z+c-d\z+\al-\be\z)]^l
\right\rangle
\right\}\right|_{x=a_1}
\nonumber
\\[10pt]
&\ \ \ \ \ \ \ \ \ \ \ \ 
-
{\mathcal D}^k\left.\left\{
\frac{1}{2\pi i}
\left\langle
  \z^{x_1^{(n)}-y_1^{(n)}+c-d-1}[(1-q\z)x+x_1^{(n)}-y_1^{(n)}\z+c-d\z)]^l
\right\rangle
\right\}\right|_{x=a_1}.
\label{egm}
\end{align}
The same arguments that led to \eqref{egi} yield that the dominant term in \eqref{egm} is
\begin{equation*}
\frac{1}{2\pi i}(l-k){l\choose k}
\left\langle
\frac{\z^{x_1^{(n)}-y_1^{(n)}+c-d-1}(\al-\be\z)(1-q\z)^k}{(x_1-y_1\z)^{k-l+1}}
\right\rangle
n^{l-k-1}.
\end{equation*}
One readily verifies that this is equal to $1/n$ times the result of applying $\al\frac{\partial}{\partial x_1}+\be\frac{\partial}{\partial y_1}$ to the right hand side of \eqref{egi}:
\begin{align}
&
\frac{1}{n}
\left(\al\frac{\partial}{\partial x_1}+\be\frac{\partial}{\partial y_1}\right)
\frac{1}{2\pi i}{l\choose k}
\left\langle \frac{\zeta^{x_1^{(n)}-y_1^{(n)}+c-d-1}(1-q\zeta)^{k}}{(x_1-\zeta y_1)^{k-l}}\right\rangle
n^{l-k}
\nonumber
\\[10pt]
&\ \ \ \ \ 
=
\frac{1}{2\pi i}(l-k){l\choose k}
\left\langle
\frac{\z^{x_1^{(n)}-y_1^{(n)}+c-d-1}(\al-\be\z)(1-q\z)^k}{(x_1-y_1\z)^{k-l+1}}
\right\rangle
n^{l-k-1}.
\label{egii}
\end{align}

In general, there are 9 cases of fixed residues modulo 3 for $x_1^{(n)}$ and $y_1^{(n)}$. Since \eqref{egii} is readily checked for any such fixed class, it follows for general $x_1^{(n)}$ and $y_1^{(n)}$.
\end{proof}

%
%


\begin{prop}
\label{tgc}  
When $\al-\be$ is a multiple of $3$, the $R\to\infty$ asymptotics of each entry in the $i$th row of the matrix $M_i$ defined by \eqref{eef} is obtained by applying the operator $\al\frac{\partial}{\partial x_1}+\be\frac{\partial}{\partial y_1}$ to the corresponding entry of the $i$th row of $\widetilde{M'}$ $($given by equations \eqref{edd}--\eqref{edf}$)$, and multiplying the result by $1/R$.

\end{prop}

\begin{proof} The $i$th row of $M_i$ is the difference between the $i$th row of $\widetilde{M'}_{\al,\be}$ and the $i$th row of $\widetilde{M'}$. By equations \eqref{eda}--\eqref{edc}, the first $2S$ entries in this row are of the form of the left hand side of \eqref{ega}, and the remaining ones are of the form of the left hand side of \eqref{egf}. The statement follows then by Lemma \ref{tga} (applied with $r_n=x_1^{(n)}-z_j^{(n)}$ and $s_n=y_1^{(n)}-w_j^{(n)}$, so that ${\partial}/{\partial u}={\partial}/{\partial x_1}$,  ${\partial}/{\partial v}={\partial}/{\partial y_1}$) and Lemma \ref{tgb}. 
\end{proof}

\section{Average first step is proportional to electric field}

Let ${\bold e}_i$ be the unit vector pointing in polar direction $-\frac{\pi}{6}+(i-1)\frac{\pi}{3}$, for $i=1,\dotsc,6$. Given a collection of holes $O_1,\dotsc,O_n$, let us focus on $O_1$ and its smallest possible displacements, as the other holes are kept fixed. Clearly, these displacements are determined by the vectors ${\bold e}_i$. Form the vector\footnote{ Recall that $O_1+(1,0)$ denotes the translation of $O_1$ by the vector $(1,0)$, etc.}
\begin{align}
&
{\bold v}=\frac{1}{6}\left\{
\frac{\hat\omega(O_1+(1,0),\dotsc,O_n)}{\hat\omega(O_1,\dotsc,O_n)}{\bold e}_1
+
\frac{\hat\omega(O_1+(0,1),\dotsc,O_n)}{\hat\omega(O_1,\dotsc,O_n)}{\bold e}_2
+
\frac{\hat\omega(O_1+(-1,1),\dotsc,O_n)}{\hat\omega(O_1,\dotsc,O_n)}{\bold e}_3
\right.
\nonumber
\\[10pt]
&\ \ 
\left.
+
\frac{\hat\omega(O_1+(-1,0),\dotsc,O_n)}{\hat\omega(O_1,\dotsc,O_n)}{\bold e}_4
+
\frac{\hat\omega(O_1+(0,-1),\dotsc,O_n)}{\hat\omega(O_1,\dotsc,O_n)}{\bold e}_5
+
\frac{\hat\omega(O_1+(1,-1),\dotsc,O_n)}{\hat\omega(O_1,\dotsc,O_n)}{\bold e}_6\right\}.
\label{eia}
\end{align}
It is the average\footnote{ In the fine mesh limit, each fraction inside the curly braces in \eqref{eia} approaches 1.} over the smallest possible displacements of the hole $O_1$, each displacement being weighted proportionally to the number of tilings compatible with the new position of $O_1$. If the hole $O_1$ ``wants'' to move in one of the nearest six positions, this is reflected by sampling uniformly at random from the set of tilings\footnote{ We denote by ${\mathcal T}(O_1,\dotsc,O_n)$ the set of tilings of the plane with holes $O_1,\dotsc,O_n$.}
\begin{align}
&
{\mathcal T}(O_1+(1,0),\dotsc,O_n)\cup{\mathcal T}(O_1+(0,1),\dotsc,O_n)\cup{\mathcal T}(O_1+(-1,1),\dotsc,O_n)
\nonumber
\\[10pt]
&\ \ 
\cup{\mathcal T}(O_1+(-1,0),\dotsc,O_n)\cup{\mathcal T}(O_1+(0,-1),\dotsc,O_n)\cup{\mathcal T}(O_1+(1,-1),\dotsc,O_n);
\label{eib}
\end{align}
the average over the observed positions of $O_1$ is then the vector ${\bold v}$ defined by equation \eqref{eia}.

Since the ${\bold e}_i$'s add up to zero, one readily gets from \eqref{eia} and \eqref{eba} that
\begin{equation}
{\bold v}=\frac{1}{6}\left(T_{1,0}\,{\bold e}_1+T_{0,1}{\bold e}_2+T_{-1,1}{\bold e}_3+T_{-1,0}\,{\bold e}_4+T_{0,-1}{\bold e}_5+T_{1,-1}{\bold e}_6\right).
\label{eic}
\end{equation}

Since ${\bold e}_3={\bold e}_2-{\bold e}_1$ and ${\bold e}_i+3=-{\bold e}_i$, $i=1,\dotsc,3$, we obtain from equation \eqref{ebbbb} that, under the assumptions of Theorem \ref{tba}, we have\footnote{ Recall that the field ${\bold T}$ is given by \eqref{ebbb}, so it is equal, up to a multiplicative constant, to the electric field determined by the charges corresponding to the holes $O_2,\dotsc,O_n$, measured at the location of the hole $O_1$.} 
\begin{align}
&
T_{1,0}\,{\bold e}_1+T_{0,1}{\bold e}_2+T_{-1,1}{\bold e}_3+T_{-1,0}\,{\bold e}_4+T_{0,-1}{\bold e}_5+T_{1,-1}{\bold e}_6
\nonumber
\\[10pt]
&\ \ \ \ \ \ \ \ \ \ \ \ \ \ \ \ \ \ \ \ \ \ 
\sim
\frac{2\q(O_1)}{R}\left\{({\bold T}\cdot{\bold e}_1)\,{\bold e}_1+({\bold T}\cdot{\bold e}_2)\,{\bold e}_2
+({\bold T}\cdot({\bold e}_2-{\bold e}_1))({\bold e}_2-{\bold e}_1)\right\}
\nonumber
\\[10pt]
&\ \ \ \ \ \ \ \ \ \ \ \ \ \ \ \ \ \ \ \ \ \ 
=
\frac{2\q(O_1)}{R}\left\{\left[2({\bold T}\cdot{\bold e}_1)-{\bold T}\cdot{\bold e}_2\right]{\bold e}_1+
\left[2({\bold T}\cdot{\bold e}_2)-{\bold T}\cdot{\bold e}_1\right]{\bold e}_2\right\}.
\label{eid}
\end{align}
Let ${\bold T}=(T_x,T_y)$ be the Cartesian coordinates of ${\bold T}$. Then, since ${\bold e}_1=\left(\frac{\sqrt{3}}{2},-\frac{1}{2}\right)$ and ${\bold e}_2=\left(\frac{\sqrt{3}}{2},\frac{1}{2}\right)$ in Cartesian coordinates, one readily checks that
\begin{equation}
\left[2({\bold T}\cdot{\bold e}_1)-{\bold T}\cdot{\bold e}_2\right]{\bold e}_1+
\left[2({\bold T}\cdot{\bold e}_2)-{\bold T}\cdot{\bold e}_1\right]{\bold e}_2
=\left(\frac32T_x,\frac32T_x\right)=\frac32{\bold T}.
\label{eie}
\end{equation}
Therefore, by \eqref{eic}--\eqref{eie}, we obtain the following result.

\begin{theo}
\label{tia}
Under the assumptions of Theorem $\ref{tba}$, if holes $O_2,\dotsc,O_n$ are fixed, and hole $O_1$ can move a unit step, with the probabilities of each step proportional to the number of tilings with $O_1$ in the new position, then the average step of $O_1$ is, in the limit of large separation between the holes, proportional to $\q(O_1)$ times the electric field determined by the other holes, measured at $O_1$.

More precisely, the asymptotics of the average displacement of $O_1$ is
\begin{equation}
{\bold v}\sim\left(\frac12\q(O_1){\bold T}\right)\frac{1}{R},\ \ \ R\to\infty.
\label{eif}
\end{equation}

\end{theo}

In fact, the same conclusion holds in a stronger sense: If we allow not just unit steps for $O_1$, but any step of length $\leq\rho$, for any fixed positive integer $\rho$, the above arguments are readily seen to extend and prove that the average over all such steps, when the probability of each step is proportional to the number of tilings with $O_1$ in the new position, is still proportional to $\q(O_1)$ times the electric field determined by the other holes, measured at $O_1$. Indeed, note that the set $\left\{(\al,\be):\al^2+\al\be+\be^2\leq\rho\right\}$ can be partitioned into subsets of size 6, with the vectors in each subset being rotations by 60 degrees of one another. Then apply the result of the above calculation to each of these subsets, and add up the resulting equations. This proves the following strengthening of the above result.

\begin{theo}
\label{tib}
For any integer $\rho>0$, if ${\bold v}_\rho$ is defined on the pattern of equation \eqref{eia} but by averaging instead over all displacements $(\al,\be)$ of $O_1$ of length $\leq\rho$, under the assumptions of Theorem $\ref{tba}$ we have 
\begin{equation}
{\bold v_\rho}\sim\left(c_{\rho}\q(O_1){\bold T}\right)\frac{1}{R},\ \ \ R\to\infty,
\label{eig}
\end{equation}
where $c_{\rho}$ is a constant depending only on $\rho$.
\end{theo}


\parindent0pt
{\it Remark $4$.} From the point of view of the parallel to physics, this can be interpreted as showing that the electrostatic force emerges as an entropic force (see \cite{Muller} for background). In \cite{ef} we proved that the lozenges line up along the electric field lines. By contrast, Theorems \ref{tia} and \ref{tib} shows that the pull on a given hole created by the fluctuating sea of dimers in the presence of the other holes acts precisely as the pull on the charge corresponding to $O_1$ of the electrostatic force  determined by the other charges.

\parindent15pt


\section{Conjectures for ${\bold F}$-field and ${\bold T}$-field in the presence of boundary}



In \cite{rangle} we presented a conjectural characterization of the correlation of gaps in general regions with boundary in the fine mesh limit. We present below a counterpart for the field ${\bold F}$ of average dimer orientations.

Recall that the field ${\bold F}$ is a discrete field defined at the center of each left-pointing unit triangle $u$ on the triangular lattice to be equal to the average orientation\footnote{ Over all lozenge tilings of the plane with a fixed collection of holes, under the uniform distribution.} of the lozenge that covers $u$ (see \cite{ef} for details). In \cite{ef} we proved that in the bulk, the scaling limit of ${\bold F}$ is equal, up to a multiplicative constant, to the 2D electric field determined by the system of charges obtained by replacing each hole $O$ by an electric charge of magnitude $\q(O)$. In other words, in the bulk, lozenges align along the electric field lines in the scaling limit. We give below the conjectural scaling limit of ${\bold F}$ in the presence of boundary.

It will be useful for the reader, as the conjectural scaling limits for the fields ${\bold F}$ and ${\bold T}$ below are compared to one another, to note that the definitions of these two fields are analogous: ${\bold T}$ records the average first step of a hole\footnote{ Either considering six possible unit steps (each weighted proportionally to the number of tilings compatible with the new position of the holes), or considering all possible steps of length at most $\rho$ for any fixed $\rho>0$; see Section 8.} as the other holes are kept fixed, while ${\bold F}$ records the average orientation of a lozenge in a given position as all holes are fixed. Given this, it is interesting to note the difference between the two conjectures (and it would be interesting to understand the deeper meaning of the difference between the two conjectured limits).

\begin{figure}[t]
\centerline{
\hfill
{\includegraphics[width=0.30\textwidth]{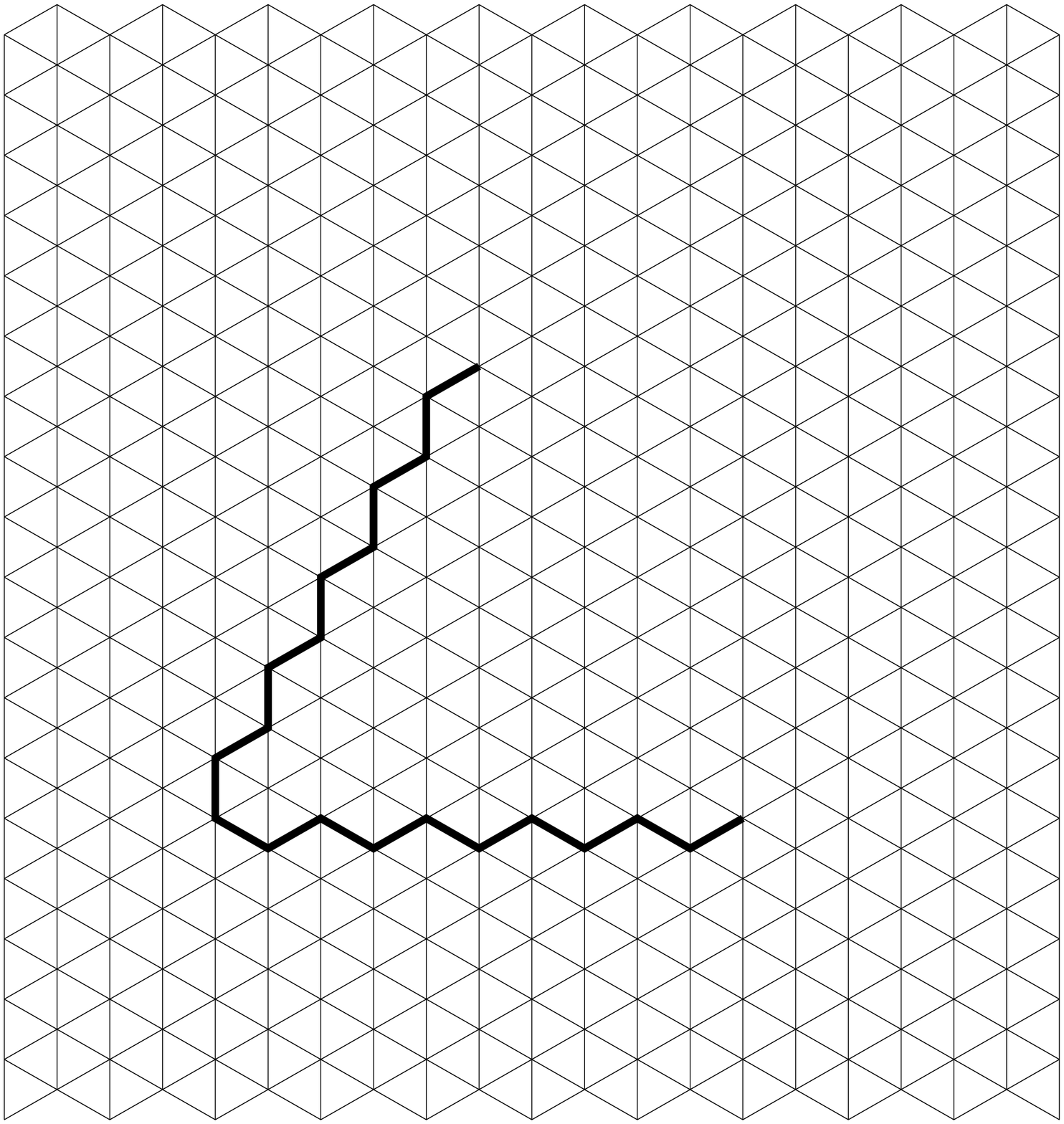}}
\hfill
{\includegraphics[width=0.30\textwidth]{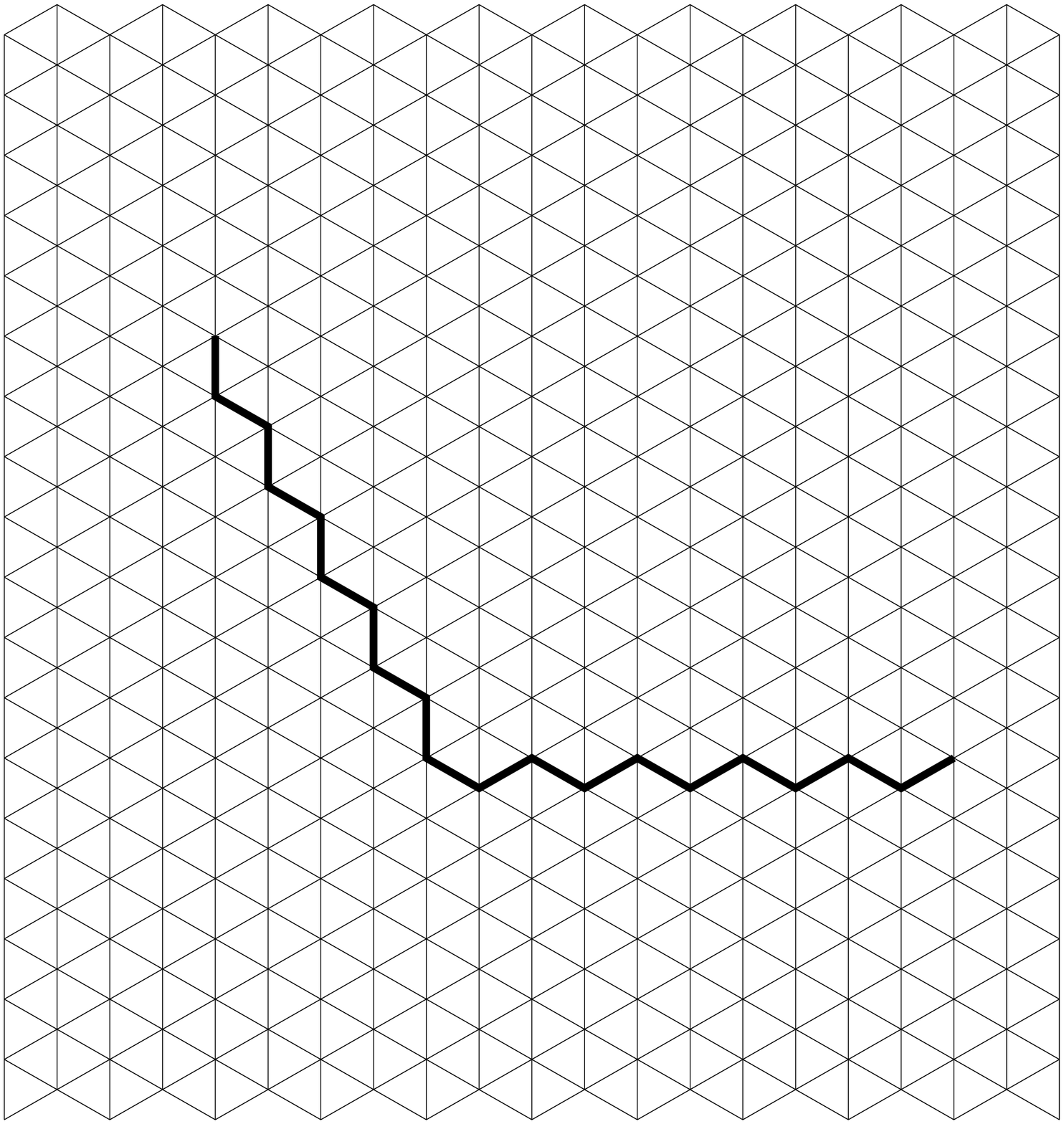}}
\hfill
}
\caption{\label{fja} The two types of zig-zag corners in $\Omega_n$.}
\end{figure}

We point out that even though the fields ${\bold F}$ and ${\bold T}$ are the same up to a multiplicative factor when there is no boundary (cf. \cite[Theorem 1.1]{ef} and Theorem \ref{tba} of the current paper), this is not true in general for regions with boundary. Indeed, in \cite{free} we showed that in a half-plane with open boundary conditions, the field ${\bold T}$ determined by a single triangular hole of side 2 has the same asymptotics as the 2D electric field near a straight line conductor, but the field ${\bold F}$ is equal to zero in the scaling limit.

In order to phrase our general conjectures, let $\Omega$ be a simply connected open set in the plane so that its boundary $\partial \Omega$ consists of a finite union of straight line segments (finite or infinite in length) with polar directions belonging to the set $\{0,\pm\pi/6,\pm\pi/3,\pm\pi/2,\pm2\pi/3,\pm5\pi/6,\pi\}$ (see Figure \ref{fjc} for an example). Color the boundary line segments with polar directions in the set $\{0,\pm\pi/3,\pm2\pi/3,\pi\}$ red (indicated by solid lines in Figure \ref{fjc}), and the boundary line segments with polar directions in the set $\{\pm\pi/6,\pm\pi/2,\pm5\pi/6\}$ blue (indicated by dashed lines in Figure \ref{fjc}). Let $a_1,\dotsc,a_k$ be distinct points in the interior of $\Omega$.

For $n\geq1$, let $\Omega_n$ be a lattice region on the triangular lattice (drawn so that one family of lattice lines is vertical) whose boundary is the union of ``straight zig-zag'' line segments (each picture in Figure \ref{fja} shows two such line segments) and lattice line segments (see Figure \ref{fjb} for an example). We assume that each corner where two zig-zag portions of the boundary meet looks, up to rotation by some multiple of $\pi/3$, as shown in Figure \ref{fja}.
Let $\Omega_n$ have constrained boundary conditions along the zig-zag portions, and free boundary conditions along the lattice segment portions. Suppose that in the scaling limit\footnote{ By the scaling limit we mean here the fine mesh limit, i.e. the limit as the lattice spacing of the triangular lattice on which the regions $\Omega_n$ reside approaches zero.}we have $\Omega_n\to\Omega$ as $n\to\infty$, in such a way that the zig-zag portions of the boundary of $\Omega_n$ approach red segments on the boundary of $\Omega$, and the free portions of the boundary of $\Omega_n$ approach blue segments on the boundary of $\Omega$.

Let $O_1^{(n)},\dotsc,O_k^{(n)}$  be finite unions of unit triangles from the interior of $\Omega_n$, so that for any fixed~$i$, the $O_i^{(n)}$'s are translates of one another for all $n\geq1$ (these will be the gaps). Assume that in the scaling limit as $n\to\infty$, $O_i^{(n)}$ shrinks to $a_i$, for $i=1,\dotsc,k$.

\begin{figure}[t]
\centerline{
\hfill
{\includegraphics[width=0.55\textwidth]{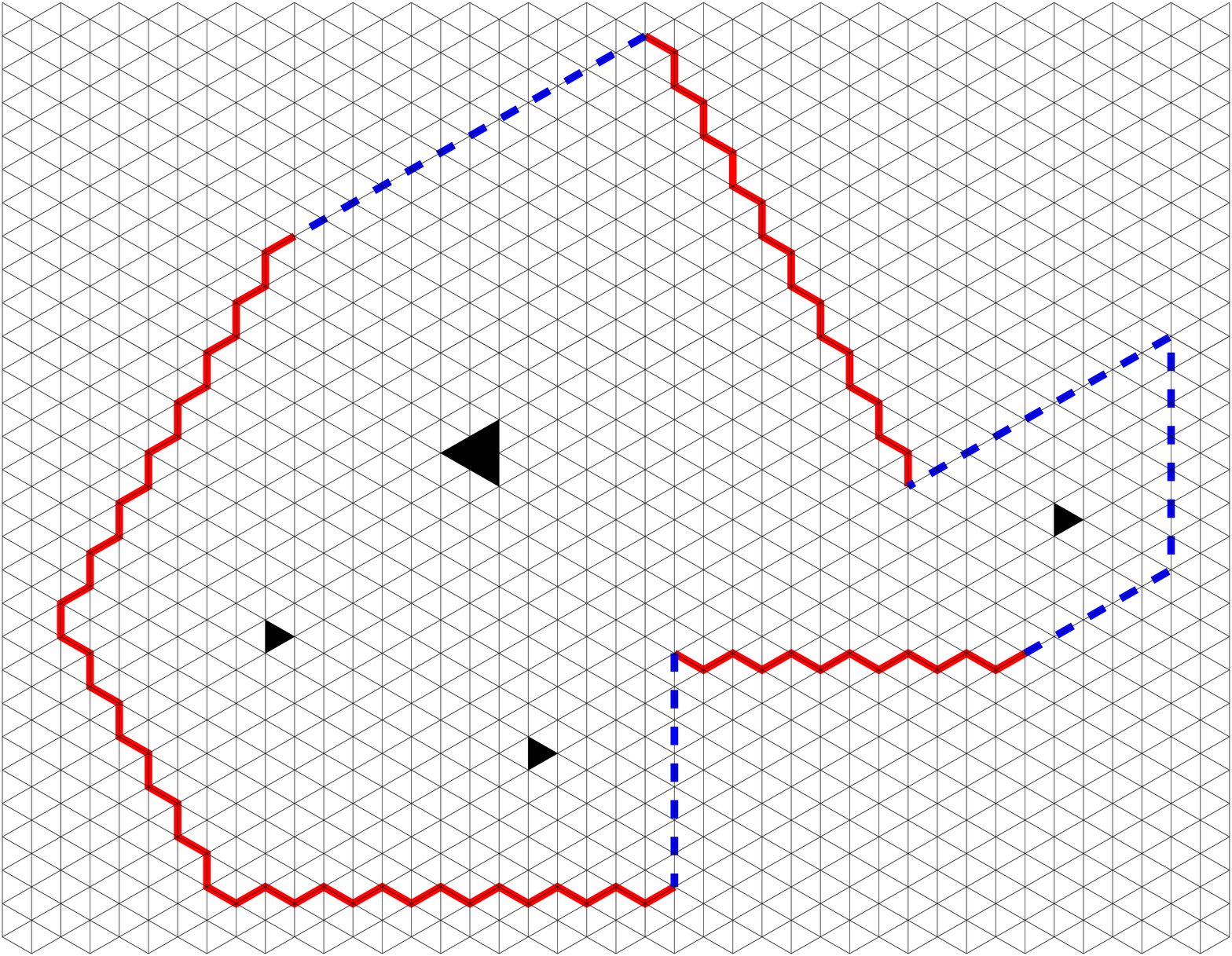}}
\hfill
}
\caption{\label{fjb}  An example of $\Omega_n$.}
\vskip0.1in
\centerline{
\hfill
{\includegraphics[width=0.55\textwidth]{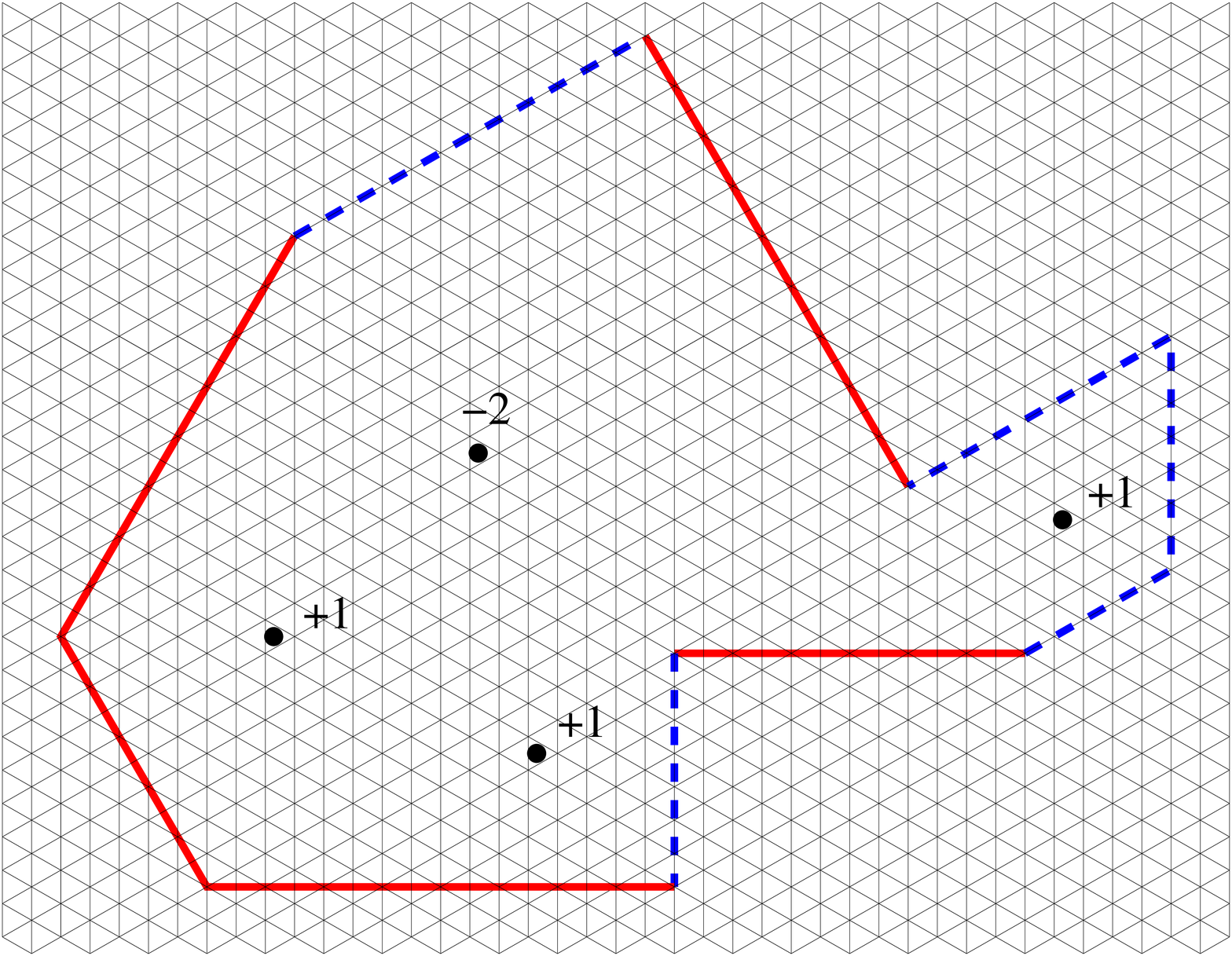}}
\hfill
}
\caption{\label{fjc} Gaps in a general region with boundary, and the corresponding steady state heat flow problem.}
\end{figure}

Corresponding to these regions, consider a 2D physical system $S$ consisting of a uniform block of material having shape $\bar\Omega$, with $k$ point sources at $a_1,\dotsc,a_k$, of strengths given by the $\q(O_i^{(n)})$'s (a heat source with a negative strength is a heat sink; see Figure \ref{fjc}). Let the boundary of this block of material be a perfect insulator along its red segments (which correspond to the constrained boundaries of $\Omega_n$), and let it be kept at a common constant temperature along its blue segments (corresponding to the free boundaries of $\Omega_n$).
Denote by ${\bold h}$ the heat flow vector field of this system in the steady state.

Then we believe that the behavior of the
field of average lozenge orientations ${\bold F}$
is given in the scaling limit by the following.

\begin{con}
\label{tja} Let the holes $O_i^{(n)}$ shrink to distinct points $a_1,\dotsc,a_k\in\Omega$ in the scaling limit as $n\to\infty$. Then
%
\begin{equation}
\lim_{n\to\infty}n\hskip0.01in{\bold F}\sim c\, {\bold h}, 
\label{eja}
\end{equation}
where $c$ is a scalar constant that depends on the region $\Omega$ $($and which portions of its boundary are perfect insulators or kept at a common temperature$)$ and on the shapes of the $O_i^{(n)}$'s, but not the value of the common temperature or the points  $a_1,\dotsc,a_k\in\Omega$ to which the $O_i^{(n)}$'s shrink in the scaling limit.
\end{con}
%

%

How about ${\bold T}$? Our considerations from \cite[Section 8]{rangle}, Theorem \ref{tba} and the example we worked out in \cite{free} suggest the following.
\begin{con}
\label{tjb} Let the holes $O_i^{(n)}$ shrink to distinct points $a_1,\dotsc,a_k\in\Omega$ in the scaling limit as $n\to\infty$. Then
%
\begin{equation}
\frac{1}{\sqrt{\alpha^2+\alpha\beta+\beta^2}}
\lim_{n\to\infty}n\,T_{\alpha,\beta}^{O_1^{(n)}}(O_1^{(n)},\dotsc,O_n^{(n)})
\sim
c\,\frac{\nabla_{(\al,\be)}E}{E},
\label{ejb}
\end{equation}
where $E$ is the heat energy\footnote{ A good way to think about the heat energy is that, up to a multiplicative constant, it equals $\int_\Omega h^2$, where $h$ is the magnitude of the heat flow vector.} of the above described system in the steady state,
$\nabla_{(\al,\be)}$ is the directional derivative\footnote{ When regarding $E$ as a function of the coordinates of $a_1$, keeping points $a_2,\dotsc,a_k$ fixed; the directional derivative of a function $f({\bold x})=f(x_1,\dotsc,x_n)$ along a vector ${\bold v}=(v_1,\dotsc,v_n)$ is the function $\nabla_{\bold v}f({\bold x})$ defined by $\nabla_{\bold v}{f}({\bold x})=\lim_{h\to0}\frac{f({\bold x}+h{\bold v})-f({\bold x})}{h|{\bold v}|}$.}, 
and $c$ is a constant that depends on the region $\Omega$ $($and which portions of its boundary are perfect insulators or kept at a common temperature$)$ and on the shapes of the $O_i^{(n)}$'s, but not on $\al$,~$\be$, the value of the common temperature or the points  $a_1,\dotsc,a_k\in\Omega$ to which the $O_i^{(n)}$'s shrink in the scaling limit.

\end{con}

Note that equation \eqref{ejb} implies that there exists a field ${\bold T}$ that governs in the scaling limit the relative changes $T_{\al,\be}$'s via equation \eqref{ebbbb}. 

In short, the above two conjectures state that, up to constant multiplicative factors, ${\bold F}$ is the gradient of temperature in the steady state heat flow problem naturally corresponding to our region with holes, and ${\bold T}$ is the logarithmic gradient of heat energy. Since steady state heat flow in the bulk has the same equations as the corresponding electrostatic field in the bulk (see e.g. \cite[\S12]{Feynman}), in the special case when there is no boundary, the statements of Conjectures \ref{eja} and \ref{ejb} reduce to the statements of \cite[Theorem 1.1]{ef} and Theorem \ref{tba} of the current paper, respectively.


\end{document}